\documentclass[11pt,a4paper,oneside]{article}
\usepackage{amsfonts,amsmath,bm,amssymb,natbib}
\usepackage{amsthm}
\theoremstyle{plain}
\usepackage{calligra}
\usepackage[dvips]{graphicx}
\usepackage{color}
\usepackage{tabu}
\newtheorem*{theorem*}{Theorem}
\newtheorem{proposition}{Proposition}
\newtheorem*{proposition*}{Proposition}

\newenvironment{my_enumerate}{
\begin{enumerate}
    \setlength{\itemsep}{1pt}
    \setlength{\parskip}{0pt}
    \setlength{\parsep}{0pt}
}{\end{enumerate}}

\begin{document}
\renewcommand{\baselinestretch}{1.5}\small\normalsize

\title
{Maximum pseudo-likelihood estimation in copula models for small weakly dependent samples}

\author
{Alexandra Dias\footnote{Email: {\tt alexandra.dias@york.ac.uk}. Address: University of York, Heslington, York, YO10 5DD, United Kingdom.}\\ {\it University of York, U.K.}
}
\date{}
\maketitle
\renewcommand{\baselinestretch}{1.6}\small\normalsize

\begin{abstract}
Maximum pseudo-likelihood (MPL) is a semiparametric estimation method often used to obtain the dependence parameters in copula models from data. It has been shown that despite being consistent, and in some cases efficient, MPL estimation can overestimate the level of dependence especially for small weakly dependent samples. We show that the MPL method uses the expected value of order statistics and we propose to use instead the median or the mode of the same order statistics. In a simulation study we compare the finite-sample performance of the proposed estimators with that of the original MPL and the inversion method estimators based on Kendall's tau and Spearman's rho. Our results indicate that the modified MPL estimators, especially the one based on the mode of the order statistics, have better finite-sample performance, while still enjoying the large-sample properties of the original MPL method. 

\vspace{0.5cm}\hspace*{-0.6cm}{\em Keywords: Copula model; Finite-sample properties; Order statistics; Relative efficiency; Semiparametric estimation.}
\end{abstract}

\newpage
\renewcommand{\baselinestretch}{1.5}\small\normalsize

\section{Introduction}
\label{sc:introd}

In 1959, Sklar formalised the concept of copula to describe the multivariate dependence structure of a random vector. Copula models are widely used in insurance and finance for pricing, hedging and risk management, as well as in health sciences, hydrology and other applied sciences; see e.g.\   \cite{rf:joe2014,rf:mcneil2015quantitative,rf:czado2019,rf:ChenGuo2019,rf:Kularatne2021}. 
Such wide applicability has triggered important contributions both in probabilistic and statistical aspects of copula models; see \cite{rf:duranteSempi2015, rf:joe2014} and references therein. The estimation of copula model parameters from observed data appears, at first, to be a straightforward inference exercise. However, it has in fact significant pitfalls. Estimation of copula models without fully understanding the properties of the estimators used can have undesirable consequences such as, among others, overestimation of the dependence in the data; see discussion in \cite{rf:fermanianScaillet2005}. One of the difficulties is that the distribution of the univariate margins is in principle unknown. Estimation procedures have been proposed to circumvent this problem, but no estimation procedure seems to be  clearly the best. In fact, \cite{rf:kojadinovicYan2010} show that the performance of the estimation method depends on the size of the sample and the strength of the dependence in the data, at least for those estimators commonly used in practice. Our goal here is to propose an estimation procedure easy to implement and that performs well from small to large samples and across a wide range of dependence levels. In line with the work from previous researchers we base the proposed estimation procedure on the following well known decomposition of a copula model into a copula function and its univariate margins.

Sklar's representation theorem  \citep{rf:sklar59} characterizes a so-called {\it copula model} for a random vector $\bm{X}=(X_1, \ldots,X_d)$ with multivariate distribution $H$ by a copula function, $C$, and univariate marginal cumulative distribution functions (cdfs) $F_i(x_i)=P(X_i\le x_i)$ for $i=1,\ldots,d$,  as
\begin{equation*}
H(\bm x) = C\left[F_1(x_1),\ldots,F_d(x_d)\right],\qquad \bm x \in \mathbb{R}^d.
\end{equation*}
A copula $C:[0,1]^d \rightarrow [0,1]$ is then a multivariate distribution with standard uniform univariate margins. If the univariate margins cdfs are continuous then the copula is unique. The versatility of copula models is apparent from Sklar's representation theorem. By combining different univariate distributions $F_i$ for the margins with copula functions $C$ a variety of models can be easily defined. Such flexibility can have a cost when it comes to the task of estimating the copula model from observed data.

Assuming that the univariate margins and copula all belong to absolutely continuous families of distributions, the obvious estimation method is maximum likelihood (ML). By default, the ML estimation of a model's parameters is done in one step. But mainly due to numerical problems, which typically arise during the optimization of a likelihood function with several parameters and possibly multi-dimensional integrals, a two-step maximum likelihood estimation method has been introduced, the so-called {\it inference functions for margins} (IFM) from \cite{rf:joeXu1996} and \cite{rf:joe1997}. The IFM method consists of estimating first the parameters for each univariate margin distribution independently, and then estimating the dependence parameters from the multivariate log-likelihood where the univariate margins parameter estimates are held fixed. Although the two-step IFM method can suffer from some loss of efficiency in cases of strong dependence, it still enjoys strong asymptotic efficiency as shown by \cite{rf:joe2005}. A fundamental challenge with the ML estimation, either one or two-step procedure, is to ensure the correct choice of distributions for the univariate margins. This is especially relevant if we are particularly interested in modelling the dependence structure of the random vector. Through a simulation study, \cite{rf:fermanianScaillet2005} find that misspecification of the margins may translate into a severe positive bias in the estimation of the copula parameters leading to an overestimation of the degree of dependence in the data. An extensive simulation study from \cite{rf:kimetal2007} show that the one-step ML and the IFM methods are indeed nonrobust against misspecification of the marginal distributions. \cite{rf:kimetal2007} also show that when the margins are unknown, in order to avoid the consequences of misspecification, it is better to use the {\it maximum pseudo-likelihood} (MPL) estimation procedure studied in \cite{rf:genestetal1995a} and \cite{rf:shihLouis1995}. 
For a random sample $\{(X_{1,i},\ldots,X_{d,i}): i=1,\ldots,n \}$ from distribution $H(\bm x)=C_{\bm\theta}\left[F_1(x_{1}), \ldots, F_d(x_{d})  \right]$, the MPL is a semiparametric estimation procedure consisting on selecting the parameter $\hat{\bm\theta}$ that maximizes the log pseudo-likelihood function
\begin{equation*}
\sum_{i=1}^n \log c_{\bm \theta}\left[F_{1,n}(X_{1,i}),\ldots,F_{d,n}(X_{d,i})\right],
\end{equation*}
where  $c_{\bm \theta}$ is the probability density function (pdf) of the copula family $\{C_{\bm \theta}\}$ and the estimated univariate marginal distributions $F_{j,n}$ are constructed from the marginal empirical distribution functions. Further asymptotic properties of the MPL estimator have been studied in \cite{rf:klaassenWellner1997} and \cite{rf:genestWerker2002}. The finite-sample properties of the MPL estimator have been studied in \cite{rf:kojadinovicYan2010} in a study where they compare the MPL estimator with the two method-of-moments (MM) estimators based on the inversion of Spearman's rho and Kendall's tau coefficients. The MM estimators have been studied by \cite{rf:oakes1982,rf:genest1987,rf:genestRivest1993}. \cite{rf:kojadinovicYan2010} find that the MPL estimator performs better than the MM estimators in terms of mean squared error except for small and weakly dependent vectors. Using the MM procedure as an alternative to MPL for small weakly dependent vectors has some problems. First, Kendall's tau and Spearman's rho inversion method estimators are mostly useful for the bivariate one parameter cases. Second, the relation between the copula dependence parameter and Kendall's tau or Spearman's rho coefficients is not always known, hence requiring the use of approximation methods in these cases. Third, as beforehand the level of dependence in unknown, a single estimator performing well across the all range of dependence level would be preferable. 

We focus on the relative disadvantage of the original MPL estimator, hereafter referred to as the canonical MPL estimator, for small weakly dependent samples.  Our goal is to improve the canonical MPL estimator, such that it outperforms the MM estimators, while keeping the MPL good large sample properties. We propose to modify the canonical MPL estimator  by using nonparametric estimators of the univariate marginal distributions different from that typically used since \cite{rf:genestetal1995a} and \cite{rf:shihLouis1995}. The MPL estimators proposed use consistent estimators of the univariate distribution functions, hence sharing the asymptotic properties of the canonical MPL. 
We study the finite-sample properties of three alternative MPL estimators, comparing their performance with the canonical MPL estimator, and with the MM based on the inversion of Kendall's tau and Spearman's rho.

Given the multitude of existing copula families and since we are comparing six estimators, we had to make choices in order to keep this study at a reasonable length. The results of \cite{rf:kojadinovicYan2010} are qualitatively similar for bivariate single parameter, multi-parameter and multivariate copula families and to some extent across different copula families. We focus this study on three bivariate copula families: the Clayton, the Gumbel-Hougaard and the Plackett. The Clayton family was first written in the form of a copula by \cite{rf:kimeldANDsamp1975}. Due to its joint lower tail dependence property, this family as been used to model the association between inter-event times, from  epidemiology to insurance. The \cite{rf:gumbel60} copula can be used to model joint upper tail dependence, for instance, between large losses on financial assets or insurance claims. The bivariate \cite{rf:plackett65} family, is derived by extending the idea of odds ratio of $2\times 2$ contingency tables to bivariate distributions with continuous margins. This family of copulas is radially symmetric and has been used as an alternative to the bivariate normal copula; see \cite{rf:nelsen06}.
Further details on each of these copula families can be found e.g.\ in \cite{rf:joe2014}. We do not wish to imply that the results will be the same as here for all other possible dimensions and copula families. It will be necessary to replicate the analysis presented here for different copula models of interest for particular applications. Studying the bivariate case is important because these models have become more relevant with the growth of the work on vine copula models \citep[see e.g.][]{rf:joe1997,AAS2009182}. Vine copulas allow to construct models more flexible than known multivariate copula families do and are based on bivariate copulas. 


 In order to compare the finite-sample properties of the estimators we perform a simulation study using the {\tt copula R package} \citep{rf:copula2020} to carry out the computations. We find that changing the nonparametric estimator of the univariate margins indeed improves the finite-sample performance of the MPL estimator, in terms of bias and mean squared error.  To confirm the large sample performance of the estimators we perform asymptotic relative efficiency calculations.

Instead of proposing to use alternative nonparametric estimators of the univariate margins, another possibility would be to obtain a bias correction function for the canonical MPL estimator. Such bias correction function would depend, not only on the copula parameter and sample size, but importantly on the specific copula itself. The approach that we propose to use here has the advantage of not depending on the specific copula. A bias reduction correction can also have the effect of increasing the variance of the estimator and possibly the mean square error \citep[see e.g.][]{rf:sorbyeEtAl2021}. That does not happen with the estimators we propose here.

Although we chose to compare the MPL estimators proposed here with the MM estimators, as in \cite{rf:kojadinovicYan2010}, other semiparametric estimators have been introduced in the literature. \cite{rf:tsukahara2005} studied two semiparametric estimation procedures and concluded that these overall have a higher mean squared error when compared with the canonical MPL estimator. \cite{rf:chenEtAl2006} introduced and studied the properties of an MPL estimator where the unknown marginal density functions are approximated by linear combinations of finite-dimensional known basis functions with increasing complexity called sieves. They find that for weak dependence the sieve method performs comparably to the canonical MPL in finite samples. Given these results, comparing the proposed estimators with the canonical MPL and the MM estimators seems an appropriate choice.

In Section~\ref{sec:cMPL} of this article we introduce the canonical MPL estimation procedure and its statistical properties. It is our starting point as we benchmark the MPL estimators that we propose against the canonical MPL estimator. Section~\ref{sec:newMPLs} contains the newly proposed MPL estimators. Section~\ref{sec:moments} summarises the MM estimators used in the simulation study. In Section~\ref{sec:simulations} we report and discuss the results of the simulation study where we compare the performance of the proposed MPL estimators with the performance of the canonical MPL and the two MM estimators. Section~\ref{sec:concl} concludes the paper. The derivation of asymptotic properties of the proposed estimators and some of the simulation results for larger samples are given in the Appendices.

\section{The canonical maximum pseudo-likelihood estimator}
\label{sec:cMPL}

Given a multivariate model with univariate marginal absolutely continuous distribution functions belonging to parametric families $F_{1,\bm\beta_1},\ldots,F_{d,\bm\beta_d}$, and absolutely continuous copula function from a parametric family $C_{\bm\theta}$, the task is to estimate the model parameters $\bm\beta_1,\ldots,\bm\beta_d,\bm\theta$ from a random sample $\bm X_1,\ldots,\bm X_n$. The classical approach is to use maximum likelihood where the parameter estimators maximize the full log-likelihood function
\[\sum_{i=1}^n \log c_{\bm\theta}[F_{1,\bm\beta_1}(X_{1,i}),\ldots,F_{d,\bm\beta_d}(X_{d,i})]+\sum_{i=1}^n\sum_{j=1}^d\log f_{j,\bm\beta_j}(X_{j,i}),\]
where $c_{\bm\theta}$ and $f_{1,\bm\beta_1},\ldots,f_{d,\bm\beta_d}$ are the probability density functions for the copula and for the marginal distributions respectively. Under the classical regularity conditions, the maximum likelihood estimator is asymptotically normal and asymptotically efficient as long as the model is well specified. The method involves estimating all the parameters simultaneously and this can be computationally challenging. The numerical optimization with many parameters might perform poorly or the presence of multi-dimensional integrals, which are numerically time consuming to evaluate, may impair a full maximum likelihood estimation procedure in practice. An alternative is to estimate the model parameters in two steps, using the so-called {\it inference functions for margins} method  \citep{rf:joeXu1996,rf:joe1997}. The method consists of first obtaining estimates $\hat{\bm \beta}_{1},\ldots,\hat{\bm\beta}_{d}$ for the marginal distributions parameters assuming that the variables are independent and then, in a second step, retrieving the copula parameter estimates by maximizing
\[\sum_{i=1}^n \log c_{\bm\theta}[F_{1,\hat{\bm\beta}_{1}}(X_{1,i}),\ldots,F_{d,\hat{\bm\beta}_{d}}(X_{d,i})]\]
in order to $\bm\theta$.
This estimator is still asymptotically normal but there can be a loss of efficiency relatively to the one-step maximum likelihood estimator \citep[see][]{rf:joe2005}. A further advantage of a two-step estimation method is that the estimation of the univariate margins parameters is not affected by a possible misspecification of the multivariate copula model.

In most applications the univariate marginal distributions are unknown and if these are misspecified  then the copula parameter estimator lacks robustness \citep{rf:kimetal2007}, it can be severely biased and have high mean squared errors \citep{rf:fermanianScaillet2005}.
Another possible approach is to estimate the multivariate model non-parametrically using empirical copulas, which asymptotic properties can be found in \cite{rf:genestsegers2010} and  \cite{rf:segersetal2017}. Naturally, empirical copula model estimation requires larger samples. Especially in applications where the size of the sample available is limited, there might be enough data to estimate the univariate marginal distribution functions non-parametrically but not enough data to estimate the empirical copula. That is one of the reasons why the semiparametric method from \cite{rf:genestetal1995a} has become commonly used. This so-called {\it maximum pseudo-likelihood} method, here called canonical MPL, consists of estimating univariate marginal distributions $\hat F_{1},\ldots,\hat F_{d}$ from the marginal empirical distributions, in a first step assuming that the univariate variables are independent, and then select the copula parameter that maximizes the log pseudo-likelihood function
\begin{equation}
\sum_{i=1}^n \log c_{\bm\theta}\left[\hat F_{1}(X_{1,i}),\ldots,\hat F_{d}(X_{d,i})\right]=\sum_{i=1}^n \log c_{\bm\theta}\left(\hat U_{1,i},\ldots,\hat U_{d,i}\right).
\label{eq:mpl}
\end{equation}
In the canonical MPL estimation, the so-called {\it pseudo-observations} $\bm{\hat U}_i=\left(\hat U_{1,i},\ldots,\hat U_{d,i}\right)$ are obtained from $\bm X_i=\left(X_{1,i},\ldots,X_{d,i}\right)$ as
\begin{equation}
\bm{\hat U}_{j,i}=\left(\frac{n}{n+1} F_{1,n}(X_{1,i}),\ldots,\frac{n}{n+1} F_{d,n}(X_{d,i})\right)\qquad i=1,\ldots,n,
\label{eq:pseudo1}
\end{equation}
 where $F_{j,n}$ is the empirical cumulative distribution function $F_{j,n}(x)=1/n\sum_{k=1}^n \bm 1(X_{j,k}\le x)$ for $j=1\ldots,d$, and $\bm 1(A)$ denotes the indicator function of event $A$. The rescaling of the empirical distribution function by the factor $n/(n+1)$ in expression (\ref{eq:pseudo1}) is made to avoid computational problems arising from the unboundedness of the log pseudo-likelihood function (\ref{eq:mpl}) on the boundary of $[0,1]^d$. The use of the empirical distribution function to transform the margins to uniform can be traced back to \cite{rf:genestetal1995a}.
The large sample properties of the canonical MPL estimator were studied by \cite{rf:genestetal1995a} and \cite{rf:shihLouis1995} who showed that this estimator is consistent and asymptotically normal, and efficient at independence. Later, \cite{rf:genestWerker2002} argue that the latter is rather the exception than the rule and identify two cases of semiparametric efficiency. These are the independence and the normal copula, for which the result could already be found in \cite{rf:klaassenWellner1997}.
 
The MPL estimation method hinges on the nonparametric estimation of $F_1,\ldots,F_d$. This is the centre of our attention hereafter in this article. We motivate the use of alternative nonparametric estimators of the univariate distribution functions $F_j$, $j=1,\ldots,d$, in Section~\ref{sec:newMPLs}, and show, in Section~\ref{sec:simulations}, that these improve the small-sample performance of the canonical MPL estimator, in terms of bias and mean squared error, while preserving its asymptotic properties.

\section{Alternative MPL estimators}
\label{sec:newMPLs}

The semiparametric canonical MPL estimation procedure relies on a nonparametric estimator of each marginal univariate  distribution $F_{j}$ for $j=1,\ldots,d$. As introduced in the previous section, this nonparametric estimator is the rescaled empirical distribution function $n/(n+1)F_{j,n}$. Here, we motivate and propose the use of alternative nonparametric estimators for the univariate margins in the MPL estimation procedure. In fact, we can use any consistent estimator for the univariate margins \citep{rf:shihLouis1995}. We show in Appendix 1 that the alternative estimators used here are consistent estimators of the univariate marginal distributions.

In the implementation of the canonical MPL method, for each univariate margin $j$ ($j=1,\ldots,d$), the pseudo-observations $\hat U_{j,1}, \ldots,\hat U_{j,n}$  defined in (\ref{eq:pseudo1}) are calculated as 
\[\hat U_{j,i} =\sum_{k=1}^n \bm{1}(X_{j.k}\le X_{j.i})/(n+1)=R_{j,i}/(n+1),\] where $R_{j,i}$ is the rank of $X_{j,i}$ among $X_{j,1},\ldots,X_{j,n}$. 
Hence $R_{j,i}/(n+1)$ can be seen as an estimator of the cdf $F_j$, and it is indeed a consistent estimator of $F_j$, as recalled in Appendix 1. 

\subsection{Pseudo-observations as moments of order statistics}

In this section, we first show that the pseudo-observations $\hat U_{j,i}=R_{j,i}/(n+1)$ are expected values of order statistics. We show this simple result below because it motivates for the estimators that we propose. 

Assume that $X_1,X_2,\ldots,X_n$ are $n$ independent identically distributed univariate random variables, each with cumulative distribution function $F$. Arrange these in ascending order of magnitude as
\[X_{(1)} \le X_{(2)} \le \ldots \le X_{(n)},\]
and call $X_{(r)}$ the $r$th order statistic, for $r=1,2,\ldots,n$.

\begin{proposition}
Consider a random sample $\left(X_1,X_2,\ldots,X_n\right)$ from a univariate distribution with continuous cdf $F$. Each pseudo-observation $U_i=\frac{R_i}{n+1}$, for $i=1,\ldots,n$, is the expected value of the order statistic $F(X)_{(R_i)}$, where $R_i$ is the rank of $X_i$ among $X_1,X_2,\ldots,X_n$. 
\end{proposition}

\begin{proof}
Let $F_{(r)}(x)$, for $r=1,2,\ldots,n$, denote the cdf of the $r$th order statistic $X_{(r)}$. It is well known \citep[see e.g.][]{rf:davidANDnagaraja2003} that
\[F_{(r)}(x)=\sum_{i=r}^n {n \choose i} F(x)^i\left[1-F(x)\right]^{n-i}.\]
If we assume that $X_i$ is continuous, denoting the pdf of $X_{(r)}$ by $f_{(r)}(x)$ we have that
\[f_{(r)}(x) = \frac{1}{B(r,n-r+1)}\;F(x)^{r-1}\left[1-F(x)\right]^{n-r}f(x),\]
where $f(x)=F'(x)$ is the pdf of $X_i$ and $B(a,b)=\int_0^1 t^{a-1}(1-t)^{b-1}dt$, for $a>0$ and $b>0$, is the beta function. 

Given that the cdf of $X_i$, for $i=1,\ldots,n$, is $F$, the random sample $\left(F(X_1),F(X_2),\ldots,F(X_n)\right)$ $=\left(U_1,\ldots,U_n\right)$ is drawn from a standard uniform distribution $U(0,1)$. 
Hence the pdf of the $r$th order statistic $F(X)_{(r)}=U_{(r)}$ has the expression
\[ f_{(r)}(u) = \frac{1}{B(r,n-r+1)}\;u^{r-1}(1-u)^{n-r}\qquad\qquad u\in(0,1), \]
and belongs to the family of beta distributions. 
The mean of the $r$th order statistic for a random sample from a standard uniform $U(0,1)$ distribution is then
\begin{equation*}
E\left[F(X)_{(r)}\right]=E\left[U_{(r)}\right]=\frac{r}{r+(n-r+1)}=\frac{r}{n+1}.
\end{equation*}
\end{proof}
Hence, the pseudo-observations computed in (\ref{eq:pseudo1}), as proposed by \cite{rf:genestetal1995a},
are in fact the expected value of the order statistics corresponding to the sample observations, i.e.,
\begin{equation}
\bm{\hat U}_i=\left(\hat U_{1,i},\ldots,\hat U_{d,i}\right) =\left(\frac{R_{1,i}}{n+1},\ldots,\frac{R_{d,i}}{n+1}\right)=\left(E\left[F(X_1)_{(R_{1,i})}\right],\ldots,E\left[F(X_d)_{(R_{d,i})}\right]\right),
\label{eq:pobs}
\end{equation}
for $i = 1, \ldots, n$.

Note that, as we are assuming that the random variable $X$ is continuous and its cdf $F$ is an increasing function, we have that the rank of $X_i$ among $X_1,\ldots,X_n$ is the same as the rank of $F(X_i)$ among $F(X_1),\ldots,F(X_n)$. We remark here that \cite{rf:claytonCuzick1985} also used expected order statistics from unit exponential distributions in the estimation of the dependence parameter of a bivariate hazards model.

At this point it is important to recall that our goal is to improve the performance of the canonical MPL which uses the pseudo-observations computed as in (\ref{eq:pobs}). With this objective in mind we explore the properties of the pseudo-observations in (\ref{eq:pobs}) inherited from the fact that these are expected values of order statistics, and how this affects the performance of the canonical MPL estimator.

If the random variable $X$ has cdf $F$, then the distribution of the order statistics $F(X)_{(r)}$ is skewed (except for $r=n/2$ if $n$ is even), especially when $r$ is closer to $1$ or $n$. Given that the expected value can be highly influenced by the skewness of the distribution, it is then possible that the properties of the pseudo-observations in (\ref{eq:pobs}) are affected by the skewness of $F(X)_{(r)}$ and consequently also the canonical MPL estimator. Figure~\ref{fg:density} displays the pdf of the order statistic $F(X)_{(45)}$ in $(0.75,1)$, for a sample of size $n=50$. The strong knewness of the pdf implies that the mean is further away from the peak of the distribution than the median and, obviously, the mode. The pseudo-observations calculated using the mean of the order statistics might suffer from the skewness of the pdf. Hence, we propose to use the median or the mode of the order statistics, instead of the mean, to compute the pseudo-observations, and we study their effect on the performance of the MPL estimator obtained from (\ref{eq:mpl}) for copula models.

\begin{figure}[htb]
\center
	\includegraphics[clip=true,viewport=0 45 370 262,scale=0.7]{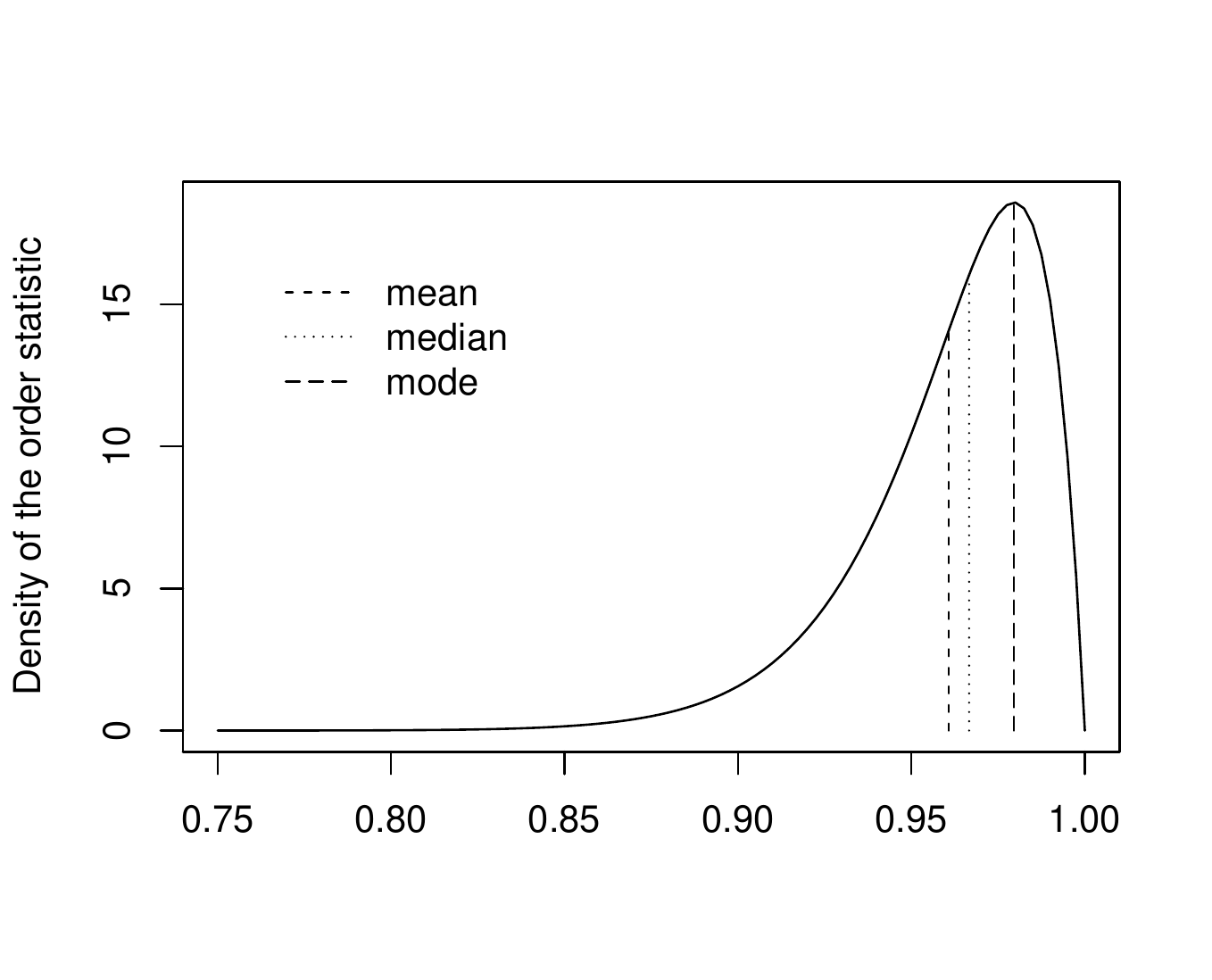}
\caption{\footnotesize Graph of the probability density function in $(0.75,1)$ of the order statistic $F(X)_{(45)}$ of a sample of size $n=50$ from a uniform $U(0,1)$ distribution. The vertical lines are the mean, median and mode of $F(X)_{(45)}$ as labeled in the legend.}
\label{fg:density}
\end{figure}

\subsection{Pseudo-observations as the median of order statistics}

We first propose to use the median of the $r$th order statistic as an alternative to using the mean of the order statistic. 
If the continuous random variable $X$ has cdf $F$ then $F(X)$ is drawn from a standard uniform distribution and the median of the order statistic $F(X)_{(r)}$ is
\[\text{med}\left(F(X)_{(r)}\right)=I^{[-1]}_{1/2}(r,n-r+1),\qquad\text{ for } 1\le r\le n,\]
where $I_p(a,b)=\int_0^p t^{a-1}(1-t)^{b-1}dt/B(a,b)$ is the regularized incomplete beta function. The computations can be made faster using the approximation \cite[see][]{rf:hyndmanFan1996,rf:kerman2011} given by 
\[\text{med}\left(F(X)_{(r)}\right)\approx \frac{r-\frac{1}{3}}{n+\frac{1}{3}},\qquad\text{ for } 1\le r\le n.\]
The corresponding pseudo-observations for the estimation of the copula parameter via the pseudo-likelihood method are then
\begin{equation}
\left(\bar U_{1,i},\ldots,\bar U_{d,i}\right)=\left(\text{med}\left[F(X_1)_{(R_{1,i})}\right],\ldots,\text{med}\left[F(X_d)_{(R_{d,i})}\right]\right) ,\qquad\text{for }1 \le i \le n.
\label{eq:median}
\end{equation}
We will refer to the copula parameter estimation procedure consisting of using the pseudo-observations given by (\ref{eq:median}) in the log pseudo-likelihood function in (\ref{eq:mpl}) as the median MPL.

\subsection{Pseudo-observations as the mode of order statistics}

The second alternative we explore to compute the pseudo-observations is using the mode of the $r$th order statistic from a standard uniform distribution, which is given by
\[\text{mode}\left(F(X)_{(r)}\right) = \frac{r-1}{n-1} \qquad\qquad\text{ for } 1 \le r \le n.\]
In this case, the pseudo-observations are computed as
\begin{align}
\left(U_{1,i}^*,\ldots,U_{d,i}^*\right) &=\left(\text{mode}\left[F(X_1)_{(R_{1,i})}\right],\ldots,\text{mode}\left[F(X_d)_{(R_{d,i})}\right]\right)\nonumber\\ &=\left(\frac{R_{1,i}-1}{n-1},\ldots,\frac{R_{d,i}-1}{n-1}\right),\qquad\qquad\qquad\text{ for }1 < i < n.
\label{eq:mode}
\end{align}
We will refer to the copula parameter estimation procedure consisting of using the pseudo-observations given by (\ref{eq:mode}) in the log pseudo-likelihood function as the mode MPL.

For the minimum and the maximum in each margin, i.e. for $X_{j(1)}$ and $X_{j(n)}$ ($j=1,\ldots,d$), it is not possible to use the mode of the corresponding order statistic as pseudo-observations in the pseudo log-likelihood function because these would be zero and one respectively. In these cases, we use instead the mean of the order statistics $1/(n+1)$ and $n/(n+1)$ as in the canonical MPL because this is our benchmark estimation procedure.

At this point we would like to remark the following. Instead of calculating the pseudo-observations as the mean, median or mode of the order statistics $F(X)_{(r)}$, we could consider using $F(E(X_{(r)}))$, $F(\text{median}(X_{(r)}))$ or $F(\text{mode}(X_{(r)}))$. If $F$ is strictly monotonic then $F(\text{median}(X_{(r)}))=\text{median}(F(X)_{(r)})$, which is one of the proposed estimators above. The pseudo-observations calculated as $F(E(X_{(r)}))$ or $F(\text{mode}(X_{(r)}))$ depend on the distribution $F$. As we want to assume that $F$ is unknown, we do not consider this possible alternative approach.

\subsection{Empirical distribution and pseudo-observations}

In the canonical MPL estimator, the motivation to rescale the empirical distribution multiplying it with $n/(n+1)$ is justified \cite[starting with][]{rf:genestetal1995a} by the need of keeping the pseudo-observations away from the boundary of the interval $(0,1)$.
The adjustment to the empirical distribution function $F_{j,n}$, in order to keep the argument of the copula density in (\ref{eq:mpl}) away from the boundary, can be done differently using
\begin{equation*}
F_{j,n}(x)-\frac{1}{2n}=\frac{1}{n}\left[\sum_{i=1}^n \bm{1}(X_{j,i} \le x)-\frac{1}{2}\right].
\end{equation*}
Here the additive factor $-\frac{1}{2n}$ ensures that the pseudo-observations are strictly in the interval $(0,1)$. This approach, introduced by \cite{rf:hazen1914}, is popular with hydrologists and it is also used by \cite{rf:joe2014} in the process of converting sample observations to normal scores. We include it in our study as an alternative to calculate the pseudo-observations, which are then given by
\begin{equation}
\left(\tilde U_{1,i},\ldots,\tilde U_{d,i}\right)=\left( \hat F_{1}\left(X_{1,i}\right),\ldots,\hat F_{d}\left(X_{d,i}\right) \right) =\left(\frac{R_{1,i}-1/2}{n},\ldots,\frac{R_{d,i}-1/2}{n}\right),\text{ for }1 \le i \le n.
\label{eq:joe}
\end{equation}
We will refer to the copula parameter estimation procedure consisting of using the pseudo-observations given by (\ref{eq:joe}) in the log pseudo-likelihood function in (\ref{eq:mpl}) as the midpoint MPL.

\subsection{Consistency of the pseudo-observations estimators}

Before moving on to the small-sample performance simulation study, we wrap up this section considering the asymptotic consistency of the different estimators. As already pointed out by other authors \cite[e.g.][]{rf:genestetal1995a,rf:kojadinovicYan2010} using $\hat U_{j,i}=R_{j,i}/(n+1)$ as pseudo-observations in the log pseudo-likelihood function in (\ref{eq:mpl}) corresponds to multiplying $n/(n+1)$ by the empirical distribution of the univariate $j$th variable.
In fact, each of the three estimators proposed above can be written as a function of the empirical distribution estimator $F_{j,n}$ of the cdf $F_j$.  We show in Appendix 1 that, as a consequence, all estimators for the pseudo-observations proposed in this study are consistent.

\section{Method-of-moments estimators}
\label{sec:moments}

In our simulation study, we also compare the performance of the four semiparametric MPL estimators with the method-of-moments estimators obtained from the relation between the copula parameter and the coefficients Kendall's tau, $\tau$, and Spearman's rho, $\rho$; see \cite{rf:oakes1982}, \cite{rf:genest1987}, \cite{rf:genestRivest1993}. Copula parameter estimates obtained from these rank coefficients via the MM can be referred to as inversion-method estimates. The reason to include the two inversion-method estimators is first, because these perform better than the canonical MPL estimator for small weakly dependent samples, and second, to facilitate the comparison of our results with other related studies.

The MM estimation procedure is mostly used in the bivariate one-parameter copula model case, although it may be used in the multivariate and/or multiparameter cases, for instance, by imposing conditions on the dependence structure. In our simulation study we restrict ourselves to the one-parameter bivariate copulas case as explained in Section~\ref{sec:simulations}. Hence, consider the random sample $\mathbf{X}_1,\ldots,\mathbf{X}_n$ from an absolutely continuous bivariate copula model $C_{\theta}\left(F_1,F_2 \right)$, where $\theta$ belongs to an open subset of $\mathbb R$, and $F_1$ and $F_2$ are continuous cdfs. Inversion-method estimators rely on a consistent estimator of a copula moment. A consistent estimator of the copula moment Kendall's tau is given by
\begin{equation*}
\tau_n = \frac{4}{n(n-1)}\sum_{i \not= j}\mathbf{1}\left(X_{1,i} \le X_{1,j}  \right)\, \mathbf{1}\left(X_{2,i} \le X_{2,j} \right)-1.
\end{equation*}
Given the ranks $\mathbf{R}_1,\ldots,\mathbf{R}_n$ corresponding to $\mathbf{X}_1,\ldots,\mathbf{X}_n$, where $R_{j,i}$ is the rank of $X_{j,i}$ among $X_{j,1},\ldots,X_{j,n}$ for $j=1,2$, a consistent estimator of the bivariate copula moment Spearman's rho is
\begin{equation*}
\rho_n = \frac{12}{n(n+1)(n-1)}\sum_{i = 1}^n R_{1,i}\,R_{2,i}-3 \frac{n+1}{n-1}.
\end{equation*}
The copula parameter estimate, $\hat\theta$, is then obtained by inversion from the relation between $\theta$ and $\tau$ or $\rho$ as $\hat \theta_{\tau}=\tau^{-1}(\tau_n)$ or as $\hat \theta_{\rho}=\rho^{-1}(\rho_n)$, when the functions $\tau$ and $\rho$ are bijections. In those cases where there is no analytic expression for the relation between the copula parameter and $\tau$ or $\rho$ then a numerical approximation must be used. The consistency, asymptotic normality and variance of  $\hat \theta_{\tau}$ and $\hat \theta_{\rho}$ are well documented in the literature and we refrain from repeating it here, directing the reader to \cite{rf:kojadinovicYan2010} and relevant references therein.

\section{The performance of the estimators}
\label{sec:simulations}

In this section we compare the performance of the semiparametric pseudo-likelihood estimator when calculating the pseudo-observations as in (\ref{eq:pobs}), (\ref{eq:median}), (\ref{eq:mode}) and (\ref{eq:joe}), and the MM Kendall's tau and Spearman's rho estimators. Recall that we refer to the MPL estimators for the copula model parameters corresponding to (\ref{eq:pobs}), (\ref{eq:median}), (\ref{eq:mode}) and (\ref{eq:joe}) as canonical MPL, median MPL, mode MPL and midpoint MPL, respectively. 

To compare the performance of the six estimators we perform a Monte Carlo study. The calculations are done using R \citep{rf:R2020} and the package {\tt copula} \citep{rf:copula2020}.

As mentioned before, we consider the Clayton, the Gumbel-Hougaard and the Plackett copulas in our simulation study. 
Let $(X_1,X_2)$ be a random vector with continuous marginal cdfs $F_1$ and $F_2$. The bivariate Clayton copula family is given by
\[ C_{\theta}\left(F_1(x_1),F_2(x_2)\right) = C_{\theta}(u_1,u_2)=\left(u_1^{-\theta}+u_2^{-\theta}-1\right)^{-1/\theta} \]
for $0 \le u_1, u_2 \le 1$ and $\theta >0$.  This family covers the full range of positive dependence, including the independence copula ($\theta\rightarrow 0$) and the Fr\'{e}chet-Hoeffding upper bound of perfect dependence ($\theta \rightarrow\infty$). We do not consider here the extension to negative dependence of the Clayton family, introduced by \cite{rf:genestMacKay1986}, because of its limited interest for statistical modelling. There is an explicit expression for the relation between Kendall's tau and the copula parameter, $\tau=\theta/(\theta+2)$, but the relation between Spearman's rho and $\theta$ has to be obtained via numerical approximation. We use the approximation implemented in the R package {\tt copula} \citep{rf:copula2020} to obtain $\hat\theta_{\rho}$ from $\rho_n$. The bivariate Gumbel copula has cdf
\[C_{\theta}(u_1,u_2)=\exp\left\{ -\left[\left(-\log u_1\right)^{\theta}+ \left(-\log u_2 \right)^{\theta}\right]^{1/\theta}  \right\}\]
for $0 \le u_1, u_2 \le 1$ and $\theta \ge 1$. The dependence of the Gumbel copula ranges from independence, when $\theta=1$, to perfect positive dependence, when $\theta \rightarrow\infty$. The coefficient Kendall's tau is given by the relation $\tau=1-1/\theta$. The Spearman's rho, as for the Clayton family, has to be obtained numerically. The Plackett copula, for $\theta >0$ and $\theta \not=1$, has cdf
\[ C_{\theta}(u_1,u_2)=\frac{1}{2\eta}\left\{ 1+\eta(u_1+u_2)-\left[(1+\eta(u_1+u_2))^2-4\theta\eta u_1 u_2      \right]^{1/2}  \right\}  \]
where $\eta=\theta-1$ and $0 \le u_1, u_2 \le 1$.  This is a comprehensive family of copulas as the limits of $C_{\theta}$ as $\theta\rightarrow 0$ and $\theta\rightarrow \infty$ correspond to the lower and upper Fr\'{e}chet-Hoeffding bounds respectively. For $\theta=1$ the Plackett copula reduces to the independence copula $C_1(u_1,u_2)=u_1 u_2$. In the case of this family, the Kendall's coefficient must be obtained by approximation while  Spearman's rho is given by $\rho = (\theta+1)/(\theta-1)-(2\theta \log \theta)/(\theta-1)^2   $.

Our main interest is on the performance of the estimators for small weakly dependent samples. With that in mind, we consider six different levels of dependence corresponding to Kendall's tau of 0.1, 0.2, 0.3, 0.4, 0.6 and 0.8, and four sample sizes of 50, 100, 200 and 400. These choices are also informed by the study of \cite{rf:kojadinovicYan2010} to make it possible to benchmark some of our results against theirs. For each level of dependence and sample size we simulate $5,000$ samples from the three copula families. Each sample is then used to estimate the copula parameter and standard error. 

\begin{table}[t]
	\renewcommand{\baselinestretch}{1.2}\small\normalsize
	\center
	\tiny
	\begin{tabular}{ccccccccccccccc}
		\hline
		$\tau$ &C$_{\theta}$& $\theta$ & $\text{\bf PRB}_{\hat\theta^c}$ & $s_{\hat\theta^c}$ & $\overline{se}_{\hat\theta^c}$ & $\text{PC}_{\hat\theta^c}$ & $\text{\bf PRB}_{\hat\theta^m}$ & $s_{\hat\theta^m}$ & $\overline{se}_{\hat\theta^m}$ & $\text{PC}_{\hat\theta^m}$ & $\text{\bf PRB}_{\hat\theta^M}$ & $s_{\hat\theta^M}$ & $\overline{se}_{\hat\theta^M}$ & $\text{PC}_{\hat\theta^M}$  \\
		\hline
		0.1 &C& 0.22 & \bf 37.8 & 0.232 & 0.240 & 97.4 & \bf 24.5 & 0.213 & 0.224 & 98.2 & \bf 15.1 & 0.200 & 0.211 & 98.9 \\
		&G& 1.11 & \bf 3.6 & 0.125 & 0.134 & 98.1 & \bf 2.3 & 0.117 & 0.128 & 98.6 & \bf 1.3 & 0.108 & 0.122 & 99.0 \\
		&P& 1.56 & \bf 11.9 & 0.799 & 0.770 & 92.1 & \bf 10.1 & 0.767 & 0.744 & 92.1 & \bf 7.4 & 0.721 & 0.701 & 91.7 \\[1ex]
		0.2 &C& 0.50 & \bf 23.7 & 0.298 & 0.289 & 94.6 & \bf 13.6 & 0.284 & 0.273 & 94.4 & \bf 5.4 & 0.267 & 0.267 & 93.6 \\
		&G& 1.25 & \bf 4.5 & 0.164 & 0.164 & 94.6 & \bf 2.5 & 0.156 & 0.159 & 93.5 & \bf 0.8 & 0.145 & 0.155 & 92.9 \\
		&P& 2.48 & \bf 12.1 & 1.226 & 1.192 & 92.8 & \bf 9.2 & 1.171 & 1.148 & 92.2 & \bf 5.1 & 1.095 & 1.079 & 91.1 \\[1ex]
		0.3 &C& 0.86 & \bf 16.5 & 0.367 & 0.361 & 94.7 & \bf 8.6 & 0.355 & 0.343 & 93.3 & \bf 1.4 & 0.337 & 0.346 & 92.4 \\
		&G& 1.43 & \bf 5.1 & 0.205 & 0.203 & 94.3 & \bf 2.6 & 0.195 & 0.197 & 93.2 & \bf 0.3 & 0.184 & 0.196 & 92.9 \\
		&P& 3.99 & \bf 11.6 & 1.901 & 1.852 & 93 & \bf 7.9 & 1.809 & 1.780 & 92.3 & \bf 2.7 & 1.69 & 1.677 & 90.5 \\[1ex]
		0.4 &C& 1.33 & \bf 11.8 & 0.463 & 0.464 & 94.5 & \bf 5.3 & 0.451 & 0.442 & 93.0 & \bf -1.0 & 0.433 & 0.457 & 91.9 \\
		&G& 1.67 & \bf 5.7 & 0.249 & 0.252 & 94.7 & \bf 2.8 & 0.238 & 0.246 & 93.8 & \bf 0.0 & 0.226 & 0.248 & 93.4 \\
		&P& 6.58 & \bf 10.5 & 3.030 & 2.940 & 92.7 & \bf 6.1 & 2.875 & 2.821 & 91.7 & \bf 0.3 & 2.693 & 2.666 & 90.0 \\[1ex]
		0.6 &C& 3.00 & \bf 3.7 & 0.796 & 0.854 & 94.5 & \bf -0.8 & 0.784 & 0.840 & 92.0 & \bf -5.9 & 0.759 & 0.879 & 92.1 \\
		&G& 2.50 & \bf 4.4 & 0.413 & 0.424 & 94.9 & \bf 1.0 & 0.398 & 0.415 & 93.5 & \bf -2.4 & 0.383 & 0.427 & 92.8 \\
		&P& 21.13 & \bf 5.4 & 8.914 & 8.792 & 92.5 & \bf 0.2 & 8.406 & 8.359 & 90.6 & \bf -5.4 & 8.033 & 7.995 & 87.5 \\[1ex]
		0.8 &C& 8.00 & \bf -3.9 & 1.803 & 2.354 & 94.5 & \bf -7.9 & 1.799 & 2.366 & 93.4 & \bf -11.3 & 1.741 & 2.517 & 93.3 \\
		&G& 5.00 & \bf -0.2 & 0.871 & 1.014 & 95.0 & \bf -3.8 & 0.853 & 0.993 & 92.0 & \bf -7.3 & 0.814 & 1.063 & 92.2 \\
		&P& 115 & \bf -8.0 & 42.733 & 45.820 & 87.8 & \bf -13 & 40.557 & 43.225 & 84.6 & \bf -15.9 & 40.121 & 42.529 & 82.1 \\
		\hline
		\hline
		$\tau$ &C$_{\theta}$& $\theta$ &  $\text{\bf PRB}_{\hat\theta^*}$ & $s_{\hat\theta^*}$ & $\overline{se}_{\hat\theta^*}$ & $\text{PC}_{\hat\theta^*}$ & $\text{\bf PRB}_{\tau}$ & $s_{\tau}$ & $\overline{se}_{\tau}$ & $\text{PC}_{\tau}$ & $\text{\bf PRB}_{\rho}$ & $s_{\rho}$ & $\overline{se}_{\rho}$ & $\text{PC}_{\rho}$ \\
		\hline
		0.1 &C& 0.22 & \bf 14.9 & 0.203 & 0.213 & 98.5 & \bf 20.8 & 0.231 & 0.246 & 99.0 & \bf 19.2 & 0.228 & 0.242 & 99.2 \\
		&G& 1.11 & \bf 1.4 & 0.111 & 0.126 & 98.7 & \bf 1.9 & 0.117 & 0.126 & 98.8 & \bf 1.8 & 0.115 & 0.120 & 98.8 \\
		&P& 1.56 & \bf 9.2 & 0.751 & 0.730 & 91.9 & \bf 10.7 & 0.787 & 0.756 & 91.8 & \bf 10.4 & 0.784 & 0.736 & 91.6 \\[1ex]
		0.2 &C& 0.50 & \bf 6.8 & 0.274 & 0.261 & 93.5 & \bf 9.0 & 0.311 & 0.306 & 94.0 & \bf 7.5 & 0.308 & 0.308 & 94.3 \\
		&G& 1.25 & \bf 1.2 & 0.150 & 0.158 & 92.3 & \bf 1.9 & 0.158 & 0.156 & 94.2 & \bf 1.6 & 0.156 & 0.149 & 93.4 \\
		&P& 2.48 & \bf 7.8 & 1.144 & 1.125 & 91.8 & \bf 11.1 & 1.247 & 1.200 & 91.9 & \bf 10.5 & 1.257 & 1.179 & 91.7 \\[1ex]
		0.3 &C& 0.86 & \bf 3.1 & 0.348 & 0.329 & 91.8 & \bf 6.2 & 0.394 & 0.382 & 93.7 & \bf 4.6 & 0.392 & 0.394 & 94.0 \\
		&G& 1.43 & \bf 1.0 & 0.189 & 0.196 & 92.1 & \bf 1.8 & 0.201 & 0.196 & 94.3 & \bf 1.4 & 0.200 & 0.186 & 92.4 \\
		&P& 3.99 & \bf 6.0 & 1.762 & 1.743 & 91.7 & \bf 11.7 & 2.028 & 1.931 & 92.4 & \bf 10.9 & 2.088 & 1.931 & 91.4 \\[1ex]
		0.4 &C& 1.33 & \bf 0.9 & 0.446 & 0.427 & 90.5 & \bf 5.6 & 0.507 & 0.485 & 93.5 & \bf 3.8 & 0.504 & 0.515 & 94.2 \\
		&G& 1.67 & \bf 0.9 & 0.231 & 0.245 & 92.7 & \bf 2.5 & 0.251 & 0.250 & 94.6 & \bf 1.8 & 0.249 & 0.233 & 93.1 \\
		&P& 6.58 & \bf 3.9 & 2.793 & 2.758 & 91.3 & \bf 12.1 & 3.394 & 3.214 & 92.1 & \bf 11.2 & 3.577 & 3.272 & 90.9 \\[1ex]
		0.6 &C& 3.00 & \bf -3.8 & 0.785 & 0.839 & 89.6 & \bf 4.9 & 0.901 & 0.848 & 93.0 & \bf 2.0 & 0.900 & 1.002 & 95.0 \\
		&G& 2.50 & \bf -1.2 & 0.389 & 0.420 & 92.2 & \bf 2.7 & 0.444 & 0.436 & 95.0 & \bf 0.8 & 0.441 & 0.395 & 91.4 \\
		&P& 21.13 & \bf -2.7 & 8.114 & 8.132 & 89.3 & \bf 14.4 & 11.534 & 10.745 & 92.9 & \bf 13.5 & 13.100 & 11.758 & 90.0 \\[1ex]
		0.8 &C& 8.00 & \bf -9.1 & 1.814 & 2.372 & 92.5 & \bf 5.6 & 2.165 & 2.002 & 93.2 & \bf -1.0 & 2.046 & 3.380 & 96.6 \\
		&G& 5.00 & \bf -6.1 & 0.848 & 1.00 & 90.1 & \bf 3.4 & 0.999 & 1.067 & 96.1 & \bf -2.3 & 0.916 & 1.112 & 92.2 \\
		&P& 115 & \bf -16.0 & 39.213 & 41.761 & 82.3 & \bf 18.8 & 72.894 & 66.481 & 93.1 & \bf 20.0 & 90.796 & 104.475 & 89.2 \\
		\hline
	\end{tabular}
	\renewcommand{\baselinestretch}{1}\small\normalsize
	\caption{\footnotesize Percentage relative bias (PRB), empirical standard deviation of the estimates ($s$), mean of the estimated standard errors ($\overline{se}$), and empirical percentage coverage (PC) of the approximate 95\% confidence interval for the dependence parameter. Estimates based on $5,000$ pseudo-random samples of size $n=50$.}
	\label{tb:bias_n50}
\end{table}

\subsection{Results on the finite-sample performance of the estimators}

For each copula and degree of dependence considered, we present in Tables~\ref{tb:bias_n50}, \ref{tb:bias_n100}, \ref{tb:bias_n200} and \ref{tb:bias_n400} the results for sample sizes 50, 100, 200 and 400 respectively. Tables~\ref{tb:bias_n200} and \ref{tb:bias_n400} can be found in Appendix~2. In the tables, the different copula models are labelled as: C for the Clayton, G for the Gumbel-Hougaard and P for the Plackett. 
For the six estimators, we report the percentage relative bias (PRB$_{\hat\theta}$), the empirical standard deviation of the estimates ($s_{\hat\theta}$), the mean of the estimated standard errors ($\overline{se}_{\hat\theta}$), and the empirical percentage coverage (PC$_{\hat\theta}$) of the approximate 95\% confidence interval for the dependence parameter calculated as $\hat\theta \pm 1.96\,se_{\hat\theta}$. In the tables we identify the results using a different subscript for each estimator. The notation for the canonical MPL is $\hat\theta^c$, for the median MPL is $\hat\theta^m$, for the mode MPL is $\hat\theta^M$, for the midpoint MPL is $\hat\theta^*$, for the MM Kendall's tau inversion is $\tau$, and for the MM Spearman's rho inversion is $\rho$. 

As already observed by \cite{rf:kojadinovicYan2010} the MM estimators have smaller relative bias than the canonical MPL  for smaller weakly dependent samples. The advantage of the MM estimators over the canonical MPL reduces when the sample size increases. The newly considered median, mode and midpoint MPL estimators have smaller bias than the canonical MPL for weakly dependent samples ($\tau \le 0.4$) across all sample sizes. The mode and the midpoint MPL estimators have lower bias than the MM estimators for weakly dependent samples especially for smaller samples. The differences between the estimators in terms of bias reduce as the sample size increases. The empirical percentage coverage does not seem very different across the six estimators not contradicting the normality assumption of the distribution of the estimators.

\begin{table}[t]
	\renewcommand{\baselinestretch}{1.2}\small\normalsize
	\center
	\tiny
	\begin{tabular}{ccccccccccccccc}
		\hline
		$\tau$ &C$_{\theta}$& $\theta$ & $\text{\bf PRB}_{\hat\theta^c}$ & $s_{\hat\theta^c}$ & $\overline{se}_{\hat\theta^c}$ & $\text{PC}_{\hat\theta^c}$ & $\text{\bf PRB}_{\hat\theta^m}$ & $s_{\hat\theta^m}$ & $\overline{se}_{\hat\theta^m}$ & $\text{PC}_{\hat\theta^m}$ & $\text{\bf PRB}_{\hat\theta^M}$ & $s_{\hat\theta^M}$ & $\overline{se}_{\hat\theta^M}$ & $\text{PC}_{\hat\theta^M}$  \\
		\hline
		0.1 &C& 0.22 & \bf 22.2 & 0.160 & 0.155 & 95.8 & \bf 14.2 & 0.151 & 0.148 & 96.3 & \bf 7.1 & 0.143 & 0.142 & 96.2 \\
		    &G& 1.11 & \bf 2.1 & 0.085 & 0.087 & 97.2 & \bf 1.3 & 0.082 & 0.085 & 97.6 & \bf 0.5 & 0.077 & 0.082 & 97.5 \\
		    &P& 1.56 & \bf 6.5 & 0.520 & 0.507 & 94.0 & \bf 5.8 & 0.509 & 0.498 & 93.9 & \bf 4.5 & 0.492 & 0.482 & 93.6 \\[1ex]
		0.2 &C& 0.50 & \bf 13.8 & 0.196 & 0.190 & 94.1 & \bf 7.7 & 0.190 & 0.183 & 94.2 & \bf 1.5 & 0.181 & 0.182 & 93.8 \\
		    &G& 1.25 & \bf 2.7 & 0.110 & 0.109 & 94.3 & \bf 1.5 & 0.107 & 0.107 & 94.1 & \bf 0.3 & 0.102 & 0.106 & 93.8 \\
		    &P& 2.48 & \bf 6.6 & 0.804 & 0.786 & 94.1 & \bf 5.3 & 0.786 & 0.771 & 93.8 & \bf 3.1 & 0.757 & 0.744 & 93.1 \\[1ex]
		0.3 &C& 0.86 & \bf 9.2 & 0.243 & 0.240 & 94.8 & \bf 4.3 & 0.238 & 0.232 & 93.4 & \bf -0.9 & 0.230 & 0.237 & 92.8 \\
		    &G& 1.43 & \bf 3.1 & 0.138 & 0.135 & 94.0 & \bf 1.6 & 0.134 & 0.133 & 93.7 & \bf 0.02 & 0.128 & 0.134 & 93.6 \\
		    &P& 3.99 & \bf 6.2 & 1.243 & 1.222 & 94.2 & \bf 4.5 & 1.213 & 1.198 & 93.8 & \bf 1.7 & 1.166 & 1.157 & 92.8 \\[1ex]
		0.4 &C& 1.33 & \bf 6.5 & 0.309 & 0.310 & 94.2 & \bf 2.6 & 0.305 & 0.299 & 92.9 & \bf -1.9 & 0.298 & 0.312 & 92.5 \\
		    &G& 1.67 & \bf 3.2 & 0.170 & 0.169 & 94.3 & \bf 1.6 & 0.166 & 0.166 & 94.2 & \bf -0.3 & 0.160 & 0.168 & 94.0 \\
		    &P& 6.58 & \bf 5.5 & 1.971 & 1.941 & 94.1 & \bf 3.4 & 1.920 & 1.901 & 93.5 & \bf 0.2 & 1.845 & 1.838 & 92.3 \\[1ex]
		0.6 &C& 3.00 & \bf 1.8 & 0.528 & 0.561 & 94.7 & \bf -0.8 & 0.524 & 0.556 & 93.1 & \bf -4.3 & 0.514 & 0.595 & 93.3 \\
		    &G& 2.50 & \bf 2.6 & 0.282 & 0.283 & 94.5 & \bf 0.5 & 0.276 & 0.280 & 93.9 & \bf -1.8 & 0.268 & 0.288 & 93.3 \\
		    &P& 21.1 & \bf 2.6 & 5.797 & 5.799 & 93.7 & \bf 0.1 & 5.618 & 5.654 & 92.8 & \bf -3.6 & 5.414 & 5.473 & 91.0 \\[1ex]
		0.8 &C& 8.00 & \bf -3.4 & 1.198 & 1.533 & 95.6 & \bf -5.4 & 1.199 & 1.587 & 94.9 & \bf -8.1 & 1.176 & 1.742 & 95.1 \\
		    &G& 5.00 & \bf -0.3 & 0.591 & 0.658 & 95.2 & \bf -2.5 & 0.583 & 0.660 & 93.9 & \bf -5.1 & 0.569 & 0.710 & 93.7 \\
		    &P& 115 & \bf -6.0 & 27.836 & 29.825 & 90.5 & \bf -8.9 & 26.916 & 28.828 & 88.5 & \bf -11.7 & 26.386 & 28.279 & 85.9 \\
		\hline
		\hline
		$\tau$ &C$_{\theta}$& $\theta$ &  $\text{\bf PRB}_{\hat\theta^*}$ & $s_{\hat\theta^*}$ & $\overline{se}_{\hat\theta^*}$ & $\text{PC}_{\hat\theta^*}$ & $\text{\bf PRB}_{\tau}$ & $s_{\tau}$ & $\overline{se}_{\tau}$ & $\text{PC}_{\tau}$ & $\text{\bf PRB}_{\rho}$ & $s_{\rho}$ & $\overline{se}_{\rho}$ & $\text{PC}_{\rho}$ \\
		\hline
		0.1 &C& 0.22 & \bf 8.5 & 0.146 & 0.143 & 95.6 & \bf 12.8 & 0.167 & 0.171 & 97.9 & \bf 11.9 & 0.165 & 0.169 & 98.2 \\
		    &G& 1.11 & \bf 0.7 & 0.079 & 0.084 & 97.5 & \bf 1.1 & 0.085 & 0.086 & 98.2 & \bf 1.1 & 0.083 & 0.083 & 98.6 \\
		    &P& 1.56 & \bf 5.4 & 0.504 & 0.493 & 93.8 & \bf 5.7 & 0.513 & 0.503 & 93.8 & \bf 5.8 & 0.515 & 0.498 & 93.6 \\[1ex]
		0.2 &C& 0.50 & \bf 3.5 & 0.186 & 0.178 & 93.1 & \bf 5.5 & 0.212 & 0.212 & 94.9 & \bf 4.6 & 0.211 & 0.213 & 94.9 \\
		    &G& 1.25 & \bf 0.8 & 0.104 & 0.107 & 93.6 & \bf 1.1 & 0.110 & 0.107 & 94.5 & \bf 0.9 & 0.110 & 0.105 & 93.7 \\
		    &P& 2.48 & \bf 4.7 & 0.777 & 0.763 & 93.5 & \bf 5.9 & 0.816 & 0.800 & 93.5 & \bf 5.7 & 0.825 & 0.797 & 93.6 \\[1ex]
		0.3 &C& 0.86 & \bf 1.0 & 0.235 & 0.225 & 92.2 & \bf 3.5 & 0.264 & 0.265 & 94.4 & \bf 2.6 & 0.266 & 0.271 & 94.9 \\
		    &G& 1.43 & \bf 0.7 & 0.131 & 0.133 & 93.3 & \bf 1.2 & 0.138 & 0.134 & 94.3 & \bf 1.0 & 0.140 & 0.133 & 92.8 \\
		    &P& 3.99 & \bf 3.7 & 1.198 & 1.185 & 93.5 & \bf 6.0 & 1.302 & 1.279 & 94.4 & \bf 5.7 & 1.348 & 1.305 & 93.7 \\[1ex]
		0.4 &C& 1.33 & \bf -0.0 & 0.303 & 0.290 & 91.2 & \bf 3.4 & 0.341 & 0.337 & 94.8 & \bf 2.4 & 0.343 & 0.351 & 95.5 \\
		    &G& 1.67 & \bf 0.5 & 0.163 & 0.166 & 93.7 & \bf 1.4 & 0.173 & 0.170 & 94.5 & \bf 1.1 & 0.174 & 0.165 & 94.3 \\
		    &P& 6.58 & \bf 2.4 & 1.894 & 1.881 & 93.2 & \bf 6.2 & 2.166 & 2.128 & 94.1 & \bf 5.8 & 2.290 & 2.214 & 93.6 \\[1ex]
		0.6 &C& 3.00 & \bf -2.6 & 0.525 & 0.559 & 91.8 & \bf 2.9 & 0.596 & 0.588 & 94.1 & \bf 1.5 & 0.617 & 0.657 & 95.3 \\
		    &G& 2.50 & \bf -0.7 & 0.273 & 0.283 & 93.3 & \bf 1.6 & 0.295 & 0.293 & 94.5 & \bf 0.7 & 0.301 & 0.283 & 92.4 \\
		    &P& 21.1 & \bf -1.3 & 5.520 & 5.577 & 92.3 & \bf 7.1 & 7.168 & 7.028 & 94.4 & \bf 6.9 & 8.341 & 7.85 & 93.0 \\[1ex]
		0.8 &C& 8.00 & \bf -6.7 & 1.209 & 1.626 & 94.6 & \bf 2.7 & 1.417 & 1.361 & 93.8 & \bf -0.3 & 1.451 & 1.871 & 97.2 \\
		    &G& 5.00 & \bf -3.9 & 0.581 & 0.682 & 93.3 & \bf 1.8 & 0.652 & 0.679 & 95.7 & \bf -1.0 & 0.655 & 0.642 & 93.5 \\
		    &P& 115 & \bf -10.6 & 26.36 & 28.271 & 86.7 & \bf 8.4 & 42.741 & 41.63 & 94.4 & \bf 10.9 & 59.905 & 54.611 & 92.6 \\
		\hline
	\end{tabular}
	\renewcommand{\baselinestretch}{1}\small\normalsize
	\caption{\footnotesize Percentage relative bias (PRB), empirical standard deviation of the estimates ($s$), mean of the estimated standard errors ($\overline{se}$), and empirical percentage coverage (PC) of the approximate 95\% confidence interval for the dependence parameter. Estimates based on $5,000$ pseudo-random samples of size $n=100$.}
	\label{tb:bias_n100}
\end{table}

\begin{table}[t]
	\renewcommand{\baselinestretch}{1.2}\small\normalsize
	\center
	\tiny
	\begin{tabular}{cccccccccccccccc}
		\hline
		$\tau$ &C$_{\theta}$& $\theta$ & \multicolumn{6}{c}{$n=50$} && \multicolumn{6}{c}{$n=100$} \\
		\cline{4-9}\cline{11-16}
		\multicolumn{3}{c}{} & RMSE& \multicolumn{5}{c}{PRE} && RMSE & \multicolumn{5}{c}{PRE} \\
		\cline{4-9}  \cline{11-16}
		& & & $\hat\theta^c$ & ${\hat\theta^c/\hat\theta^m}$ & ${\hat\theta^c/\hat\theta^M}$ & ${\hat\theta^c/\hat\theta^*}$ & ${\hat\theta^c/\tau}$ & ${\hat\theta^c/\rho}$ && $\hat\theta^c$ & ${\hat\theta^c/\hat\theta^m}$ & ${\hat\theta^c/\hat\theta^M}$& ${\hat\theta^c/\hat\theta^*}$ & ${\hat\theta^c/\tau}$ & ${\hat\theta^c/\rho}$ \\
		\hline
		0.1 &C& 0.22 & 0.24 & 125.6 & \bf 148.5 & 143.7 & 109.4 & 113.0 && 0.16 & 117.6 & \bf 134.8 & 128.7 & 97.6 & 99.8 \\
		    &G& 1.11 & 0.13 & 121.5 & \bf 145.8 & 137.3 & 121.4 & 126.3 && 0.08 & 113.6 & \bf 130.0 & 122.7 & 107.0 & 110.2 \\
		    &P& 1.56 & 0.82 & 109.9 & \bf 126.2 & 115.1 & 104.1 & 105.1 && 0.53 & 104.7 & \bf 113.3 & 107.2 & 103.3 & 102.7 \\[1ex]
		0.2 &C& 0.50 & 0.32 & 121.1 & \bf 142.9 & 134.8 & 104.2 & 107.1 && 0.20 & 115.1 & \bf 130.5 & 123.8 & 94.7 & 95.7 \\
		    &G& 1.25 & 0.17 & 119.9 & \bf 142.5 & 132.8 & 118.7 & 121.8 && 0.11 & 113.3 & \bf 128.1 & 121.2 & 108.5 & 108.7 \\
		    &P& 2.48 & 1.26 & 111.9 & \bf 131.0 & 118.4 & 97.6 & 96.6 && 0.82 & 105.9 & \bf 116.2 & 109.0 & 97.9 & 96.1 \\[1ex]
		0.3 &C& 0.86 & 0.39 & 118.1 & \bf 136.4 & 127.6 & 97.8 & 100.0 && 0.25 & 112.5 & \bf 123.1 & 117.8 & 92.1 & 91.8 \\
		    &G& 1.43 & 0.21 & 120.0 & \bf 139.2 & 132.1 & 114.6 & 117.0 && 0.14 & 113.5 & \bf 126.7 & 120.9 & 107.8 & 105.0 \\
		    &P& 3.99 & 1.96 & 113.6 & \bf 133.4 & 121.0 & 88.4 & 84.2 && 1.27 & 106.8 & \bf 117.7 & 110.3 & 91.6 & 85.9 \\[1ex]
		0.4 &C& 1.33 & 0.49 & 114.8 & \bf 127.5 & 120.5 & 91.3 & 93.5 && 0.32 & 109.8 & \bf 115.6 & 112.6 & 87.3 & 86.7 \\
		    &G& 1.67 & 0.26 & 120.7 & \bf 139.2 & 132.3 & 109.3 & 112.5 && 0.17 & 113.4 & \bf 124.5 & 120.0 & 104.5 & 104.3 \\
		    &P& 6.58 & 3.11 & 114.7 & \bf 133.2 & 122.8 & 79.5 & 72.4 && 2.00 & 107.4 & \bf 118.0 & 111.1 & 82.7 & 74.4 \\[1ex]
		0.6 &C& 3.00 & 0.80 & 104.8 & \bf 106.3 & 102.7 & 77.4 & 79.3 && 0.53 & \bf 102.3 & 100.0 & 99.9 & 77.6 & 73.6 \\
		    &G& 2.50 & 0.42 & 114.9 & \bf 121.4 & 119.9 & 90.6 & 93.8 && 0.28 & 109.5 & \bf 113.3 & 111.9 & 94.3 & 92.0 \\
		    &P& 21.13 & 8.99 & 114.3 & \bf 122.7 & 122.1 & 56.8 & 44.9 && 5.82 & 107.4 & \bf 113.4 & 111 & 63.2 & 47.3 \\[1ex]
		0.8 &C& 8.00 & 1.83 & \bf 93.9 & 87.0 & 87.5 & 68.6 & 80.0 && 1.22 & \bf 92.7 & 83.8 & 86.2 & 73.4 & 71.7 \\
		    &G& 5.00 & 0.87 & \bf 99.3 & 95.2 & 93.2 & 73.9 & 89.0 && 0.59 & \bf 98.0 & 89.6 & 92.7 & 80.6 & 80.9 \\
		    &P& 115 & 43.7 & \bf 102.2 & 98.2 & 101.7 & 33.1 & 21.8 && 28.7 & \bf 99.2 & 93.7 & 97.5 & 42.8 & 21.9 \\
		\hline\hline
		$\tau$ &C$_{\theta}$ & $\theta$ & \multicolumn{6}{c}{$n=200$} && \multicolumn{6}{c}{$n=400$} \\
		\cline{4-9}\cline{11-16}
		\multicolumn{3}{c}{} & RMSE& \multicolumn{5}{c}{PRE} && RMSE & \multicolumn{5}{c}{PRE} \\
		\cline{4-9}  \cline{11-16}
		& & & $\hat\theta^c$ & ${\hat\theta^c/\hat\theta^m}$ & ${\hat\theta^c/\hat\theta^M}$ & ${\hat\theta^c/\hat\theta^*}$ & ${\hat\theta^c/\tau}$ & ${\hat\theta^c/\rho}$ && $\hat\theta^c$ & ${\hat\theta^c/\hat\theta^m}$ & ${\hat\theta^c/\hat\theta^M}$& ${\hat\theta^c/\hat\theta^*}$ & ${\hat\theta^c/\tau}$ & ${\hat\theta^c/\rho}$ \\
		\hline		
		0.1 &C& 0.22 & 0.10 & 110.4 & \bf 120.9 & 116.3 & 82.7 & 84.1 && 0.07 & 107.0 & \bf 114.5 & 111.2 & 79.6 & 80.7 \\
		    &G& 1.11 & 0.05 & 107.8 & \bf 117.4 & 112.3 & 95.6 & 98.7 && 0.03 & 104.9 & \bf 110.7 & 107.5 & 89.1 & 92.1 \\
		    &P& 1.56 & 0.35 & 102.1 & \bf 106.1 & 103.2 & 100.8 & 98.5 && 0.23 & 101.1 & \bf 103.1 & 101.6 & 101.6 & 98.0 \\[1ex]
		0.2 &C& 0.50 & 0.13 & 109.0 & \bf 116.6 & 113.2 & 85.9 & 85.7 && 0.09 & 106.3 & \bf 111.3 & 109.1 & 84.7 & 84.2 \\
		    &G& 1.25 & 0.07 & 107.7 & \bf 115.7 & 111.4 & 101.9 & 101.7 && 0.05 & 105.1 & \bf 109.9 & 107.5 & 97.8 & 98.2 \\
		    &P& 2.48 & 0.54 & 102.6 & \bf 107.3 & 103.9 & 93.5 & 91.2 && 0.36 & 101.4 & \bf 103.9 & 102.0 & 92.4 & 91.0 \\[1ex]
		0.3 &C& 0.86 & 0.17 & 106.3 & \bf 109.1 & 107.8 & 87.1 & 84.8 && 0.11 & 104.2 & \bf 105.6 & 105.1 & 86.3 & 83.4 \\
		    &G& 1.43 & 0.09 & 108.2 & \bf 114.8 & 111.8 & 101.6 & 95.8 && 0.06 & 104.9 & \bf 108.3 & 107.0 & 99.1 & 92.0 \\
		    &P& 3.99 & 0.83 & 103.0 & \bf 107.8 & 104.4 & 89.3 & 81.5 && 0.57 & 101.6 & \bf 104.3 & 102.3 & 89.7 & 81.5 \\[1ex]
		0.4 &C& 1.33 & 0.21 & 104.6 & 104.6 & \bf 105.0 & 83.5 & 80.0 && 0.15 & 103.2 & 102.7 & \bf 103.5 & 82.6 & 78.2 \\
		    &G& 1.67 & 0.11 & 107.4 & \bf 111.8 & 110.0 & 99.0 & 96.2 && 0.08 & 104.4 & \bf 106.1 & 106.0 & 98.4 & 95.4 \\
		    &P& 6.58 & 1.32 & 103.1 & \bf 107.7 & 104.6 & 80.2 & 70.0 && 0.91 & 101.7 & \bf 104.4 & 102.5 & 81.2 & 71.0 \\[1ex]
		0.6 &C& 3.00 & 0.37 & \bf 99.4 & 94.4 & 96.7 & 79.4 & 72.1 && 0.26 & \bf 99.1 & 94.4 & 97.1 & 83.7 & 74.3 \\
		    &G& 2.50 & 0.19 & \bf 104.8 & 104.1 & 105.0 & 92.2 & 85.6 && 0.13 & 102.9 & 101.3 & \bf 103.1 & 93.6 & 86.9 \\
		    &P& 21.13 & 3.9 & 102.9 & \bf 105.1 & 104.1 & 64 & 46.3 && 2.72 & 101.6 & \bf 103.3 & 102.3 & 66.7 & 48.6 \\[1ex]
		0.8 &C& 8.00 & 0.87 & \bf 93.3 & 83.9 & 88.0 & 80.5 & 70.3 && 0.62 & \bf 95.3 & 87.7 & 91.5 & 83.8 & 71.1 \\
		    &G& 5.00 & 0.40 & \bf 97.2 & 87.8 & 92.7 & 82.5 & 74.7 && 0.28 & \bf 97.9 & 89.7 & 94.7 & 87.0 & 74.9 \\
		    &P& 115 & 20.07 & \bf 98.2 & 92.8 & 96.7 & 48.5 & 22.7 && 14.19 & \bf 99.1 & 95.6 & 98.4 & 52.2 & 24.5 \\
		\hline
	\end{tabular}
	\renewcommand{\baselinestretch}{1}\small\normalsize
	\caption{\footnotesize Estimated root mean square error (RMSE) of the canonical MPL estimator and percentage relative efficiency (PRE) of the median MPL, mode MPL, midpoint MPL, MM Kendall's and MM Spearman's estimators in relation to the canonical MPL estimator.}
	\label{tb:relEff}
\end{table}

Table~\ref{tb:relEff} contains the estimated root mean square error (RMSE) for the canonical MPL estimator obtained for each sample size, copula and level of dependence considered. The RMSE increases with the level of dependence and decreases as the sample becomes larger. Hence, the higher RMSE for the canonical MPL estimator is observed for small strongly dependent samples and the lower RMSE is obtained from weakly dependent large samples.
In Table~\ref{tb:relEff} we also report the percentage relative efficiency (PRE) calculated as 100 times the estimated RMSE of the canonical MPL divided by the estimated RMSE of each of the other five estimators. We observe that the MM Kendall's tau and Spearman's rho based estimators outperform the canonical MPL estimator for small weakly dependent samples but this advantage vanishes when the sample size increases or the level of dependence becomes stronger. These results are perfectly in line with the results from \cite{rf:kojadinovicYan2010}. The  three semiparametric MPL estimators proposed here outperform, in terms of MSE, both MM estimators for all levels of dependence and sample size. Consequently, the three estimators introduced also outperform the canonical MPL for small weakly dependent samples. For stronger levels of dependence, $\tau \ge 0.6$, and samples larger than $100$ the canonical MPL has the smallest MSE for the sample sizes considered. It is worth noting that in the simulations the proposed estimators substantially outperform the canonical MPL for weak dependence while for stronger dependence the outperformance of the canonical MPL is modest. Between the three MPL estimators introduced here, the mode MPL is overall the best for weakly dependent samples $(\theta \le 0.4)$. In summary, the simulations suggest that the mode MPL is overall preferable in the cases considered here.

\begin{table}[t]
	\renewcommand{\baselinestretch}{1.3}\small\normalsize
	\center
	\tiny
	\begin{tabular}{cccccccc}
		\hline
		$\tau$ & $C_{\theta}$ & $\theta$&\multicolumn{5}{c}{PRE for $n\rightarrow \infty$} \\
		\cline{4-8}
		\multicolumn{3}{c}{} & ${\hat\theta^c/\hat\theta^m}$ & ${\hat\theta^c/\hat\theta^M}$& ${\hat\theta^c/\hat\theta^*}$ & ${\hat\theta^c/\tau}$ & ${\hat\theta^c/\rho}$  \\
		\hline
		0.1 &C& 0.22 & 100.1 & 100.0 & 100.0 & 69.9 & 70.9 \\
		    &G& 1.11 &  99.7 & 100.0 & 100.0 & 84.3 & 87.6 \\
		    &P& 1.56 & 100.0 & 100.0 & 100.0 & 102.5 & 97.5 \\[1ex]
		0.2 &C& 0.50 & 100.3 &  99.9 & 100.1 & 77.8 & 77.0 \\
		    &G& 1.25 &  99.5 & 100.0 & 100.0 & 92.8 & 93.3 \\
		    &P& 2.48 & 100.0 & 100.0 & 100.0 & 91.5 & 91.6 \\[1ex]
		0.3 &C& 0.86 &  99.9 &  99.9 & 100.1 & 82.7 & 79.3 \\
		    &G& 1.43 &  99.5 & 100.0 & 100.0 & 95.5 & 86.4 \\
		    &P& 3.99 & 100.0 & 100.0 & 100.0 & 91.5 & 83.1 \\[1ex]
		0.4 &C& 1.33 & 100.1 & 100.0 & 100.1 & 84.2 & 78.2 \\
		    &G& 1.67 &  99.4 & 100.0 & 100.0 & 95.6 & 90.7 \\
		    &P& 6.58 & 100.0 & 100.0 & 100.0 & 82.9 & 73.1\\[1ex]
		0.6 &C& 3.00 & 100.4 & 100.0 & 100.0 & 81.6 & 69.8 \\
		    &G& 2.50 &  99.6 & 100.0 & 100.0 & 92.1 & 78.4 \\
		    &P& 21.13 & 100.0 & 100.0 & 100.0 & 69.8 & 50.7 \\[1ex]
		0.8 &C& 8.00 & 101.0 & 100.0 & 100.0 & 78.6 & 61.8 \\
		    &G& 5.00 &  99.8 & 100.0 & 100.0 & 87.2 & 61.6 \\
		    &P& 115 & 100.0 & 100.0 & 100.0 & 53.9 & 24.9 \\
		\hline
	\end{tabular}
	\renewcommand{\baselinestretch}{1}\small\normalsize
	\caption{\footnotesize Asymptotic percentage relative efficiency.}
	\label{tbl:asympPRE}
\end{table}

Finally, we estimated the asymptotic relative efficiency of the median MPL, mode MPL, midpoint MPL, MM Kendall's tau and MM Spearman's rho estimators in relation to the canonical MPL estimator. The asymptotic percentage relative efficiency for each estimator is calculated as the estimated variance of the canonical MPL estimate, divided by the estimated variance of the MPL estimate given by the method being compared with, multiplied by 100. The estimates are obtained from a pseudo-randomly generated sample of size $n=100,000$. The results, presented in Table~\ref{tbl:asympPRE}, confirm that the three proposed MPL estimators and the canonical MPL estimator are asymptotically equally efficient outperforming the MM Kendall's tau and Spearman's rho based estimators.

\section{Conclusion}
\label{sec:concl}

\cite{rf:kimetal2007} and \cite{rf:fermanianScaillet2005} found that misspecification of the margins leads to non robust estimation of the dependence structure of a random vector with overestimation of the degree of dependence. Later, \cite{rf:kojadinovicYan2010} find that overestimation of the degree of dependence can happen even when the unknown margins are estimated non-parametrically, especially for small weakly dependent samples, and that the mean squared error increases with the degree of dependence.

We show here that the pseudo-observations used in the canonical MPL estimation method \citep{rf:genestetal1995a} can be seen as expected values of the order statistics. Our simulation study shows that using the mode of the order statistics instead of the mean when calculating the pseudo-observations seems to reduce the overestimation of the level of dependence, outperforming the canonical MPL and the inversion methods Kendall's tau and Spearman's rho in terms of mean squared error for weakly dependent samples. For larger strongly dependent samples the canonical MPL still slightly  outperforms the proposed modified MPL estimators. Hence, within the conditions considered, our study strongly suggests that it is preferable to use the MPL estimator where the pseudo-observations are calculated as the mode of the order statistics rather than the mean.  

\section*{Acknowledgments}

We thank Professor Alexander J. McNeil (University of York, UK) for his insightful comments and stimulating discussions which improved the quality of this research.

\break
\section*{Appendix 1: Consistency of the pseudo-observations estimators} 
\label{app:1}

Let $X_1,X_2,\ldots,X_n$ be independent and identically distributed (iid) random variables from an unknown distribution with cdf $F$. The empirical distribution function can be defined  \citep[see][page 265]{rf:vanderVaart1998} as \[ F_n(x)=\frac{1}{n}\sum_{i=1}^{n} \bm{1}(X_i \le x),\qquad x \in \mathbb{R} \]
where $\mathbf 1(A)$ is the indicator function of event $A$. 

\begin{proposition*}[\bf A.1]
	$F_n(x)$ is a consistent estimator of $F(x)$ for every $n$.
\end{proposition*}

\begin{proof}
For a fixed $x$, $\mathbf 1(X_i \le x)$ is a Bernoulli random variable with parameter $p=F(x)$. Hence, $\sum_{i=1}^{n}\mathbf 1(X_i \le x)=n F_n(x)$ is a binomial random variable with  parameters $n$ and $p$, mean $np=nF(x)$ and variance $np(1-p)=n F(x)[1-F(x)]$.

The mean of $F_n(x)$ is then
\[E[F_n(x)]=\frac{1}{n}E[n F_n(x)]=\frac{1}{n}nF(x)=F(x),\]
which means that $F_n(x)$ is an unbiased estimator of $F(x)$ for each $n$, and its variance,
\[ var[F_n(x)]=\frac{1}{n^2}var[n F_n(x)]=\frac{1}{n^2}\,nF(x)[1-F(x)]=\frac{1}{n}F(x)[1-F(x)],\]
goes to zero as $n\rightarrow\infty$ for all $F(x)$.
By the strong law of large numbers, for every value of $x$, the estimator $F_n(x)$ converges to $F(x)$ as $n\rightarrow\infty$ almost surely,
thus $F_n(x)$ is a consistent estimator of $F(x)$ for every $n$  \citep[see ][page 55]{rf:lehmannCasella1998}.

\end{proof}

The following theorem is a useful result for establishing the consistency of estimators \citep[see][page 469]{rf:casellaBerger2002}. 
\begin{theorem*}[\bf A.2]
	Let $T_n$ be a consistent estimator for every $n$ of a parameter $\theta$. Let $a_1,a_2,\ldots$ and $b_1,b_2,\ldots$ be sequences of constants satisfying
	\begin{my_enumerate}
		\item $\lim_{n\rightarrow\infty}a_n=1$,
		\item $\lim_{n\rightarrow\infty}b_n=0$.
	\end{my_enumerate}
Then $U_n=a_n\,T_n+b_n$ is a consistent estimator of $\theta$ for every $n$.
\label{th:linearComb}
\end{theorem*}

Now we can easily show that the pseudo-observations calculated as in (\ref{eq:pobs}), (\ref{eq:median}), (\ref{eq:mode}) and (\ref{eq:joe}) are obtained from consistent estimators of the cdf $F$.

\begin{proposition*}[\bf A.3] Let $X_1,\ldots X_n$ be iid random variables from distribution $F$ then, for  $x\in \mathbb{R}$,
	\begin{my_enumerate}
		\item $F_{n_1}(x)=\frac{1}{n+1}\sum_{i=1}^{n} \mathbf 1(X_i\le x)$,
		\item $F_{n_2}(x)=\frac{1}{n+1/3}\left(\sum_{i=1}^{n} \mathbf 1(X_i\le x)-\frac{1}{3}\right)$,
		\item $F_{n_3}(x)=\frac{1}{n-1}\left(\sum_{i=1}^{n} \mathbf 1(X_i\le x)-1\right)$ and
		\item $F_{n_4}(x)=\frac{1}{n}\left(\sum_{i=1}^{n} \mathbf 1(X_i\le x)-\frac{1}{2}\right)$
	\end{my_enumerate}
are all consistent estimators of $F(x)$ for every n.
\end{proposition*}	

\begin{proof}
	
1. 
We can write the estimator $F_{n_1}$ as a function of the empirical distribution function as
\[ \frac{1}{n+1}\sum_{i=1}^{n} \mathbf 1(X_i\le x)=\frac{n}{n+1} F_n(x). \]
Since $n/(n+1)$ approaches 1 as $n$ goes to infinity, using Theorem A.1,
 we have that the mean of the order statistic is a consistent estimator of $F(x)$.

2. 
We can write the estimator $F_{n_2}$ as a function of the empirical distribution function as
\[ \frac{1}{n+1/3}\left(\sum_{i=1}^{n} \mathbf 1(X_i\le x)-\frac{1}{3}\right)=\frac{n}{n+1/3} F_n(x)-\frac{1/3}{n+1/3}. \]
Since $n/(n+1/3)$ approaches 1 and $\frac{1/3}{n+1/3}$ approaches zero, as $n$ goes to infinity, using theorem 
A.1, we conclude that the approximate median of the order statistic is a consistent estimator of $F(x)$.

3. 
We can write the estimator $F_{n_3}$ as a function of the empirical distribution function as
\[ \frac{1}{n-1}\left(\sum_{i=1}^{n} \mathbf 1(X_i\le x)-1\right)=\frac{n}{n-1} F_n(x)-\frac{1}{n-1}. \]
Since $n/(n-1)$ approaches 1 and $1/(n-1)$ approaches zero, as $n$ goes to infinity, using theorem A.1,
we obtain that the approximate median of the order statistic is a consistent estimator of $F(x)$.

4. 
We can write the estimator $F_{n_4}$ as a function of the empirical distribution function as
\[ \frac{1}{n}\left(\sum_{i=1}^{n} \mathbf 1(X_i\le x)-\frac{1}{2}\right)=F_n(x)-\frac{1}{2n}. \]
Since $1/(2n)$ approaches zero, as $n$ goes to infinity, using theorem A.1, 
we have that estimator (\ref{eq:joe}) is a consistent estimator of $F(x)$.
\end{proof}

\break
\section*{Appendix 2: Finite-sample performance for $n=200$ and $n=400$}
\label{app:2}

\begin{table}[htb]
	\renewcommand{\baselinestretch}{1.2}\small\normalsize
	\center
	\tiny
	\begin{tabular}{ccccccccccccccc}
		\hline
		$\tau$ & $\theta$& $\theta$ & $\text{\bf PRB}_{\hat\theta^c}$ & $s_{\hat\theta^c}$ & $\overline{se}_{\hat\theta^c}$ & $\text{PC}_{\hat\theta^c}$ & $\text{\bf PRB}_{\hat\theta^m}$ & $s_{\hat\theta^m}$ & $\overline{se}_{\hat\theta^m}$ & $\text{PC}_{\hat\theta^m}$ & $\text{\bf PRB}_{\hat\theta^M}$ & $s_{\hat\theta^M}$ & $\overline{se}_{\hat\theta^M}$ & $\text{PC}_{\hat\theta^M}$  \\
		\hline
		0.1 &C& 0.22 & \bf 9.2 & 0.106 & 0.103 & 94.1 & \bf 4.5 & 0.102 & 0.100 & 94.0 & \bf -0.4 & 0.098 & 0.097 & 93.5 \\
		&G& 1.11 & \bf 1.0 & 0.056 & 0.058 & 95.1 & \bf 0.5 & 0.055 & 0.058 & 94.8 & \bf 0.02 & 0.053 & 0.056 & 94.5 \\
		&P& 1.56 & \bf 2.1 & 0.345 & 0.340 & 94.0 & \bf 1.8 & 0.342 & 0.337 & 94.0 & \bf 1.2 & 0.336 & 0.332 & 93.8 \\[1ex]
		0.2 &C& 0.50 & \bf 6.6 & 0.133 & 0.129 & 94.2 & \bf 2.9 & 0.131 & 0.126 & 93.9 & \bf -1.2 & 0.127 & 0.126 & 93.4 \\
		&G& 1.25 & \bf 1.3 & 0.072 & 0.074 & 95.4 & \bf 0.6 & 0.071 & 0.073 & 95.2 & \bf -0.1 & 0.069 & 0.073 & 94.8 \\
		&P& 2.48 & \bf 2.1 & 0.533 & 0.528 & 94.1 & \bf 1.5 & 0.528 & 0.523 & 94.0 & \bf 0.4 & 0.517 & 0.513 & 93.5 \\[1ex]
		0.3 &C& 0.86 & \bf 4.1 & 0.169 & 0.165 & 94.0 & \bf 1.2 & 0.167 & 0.161 & 92.9 & \bf -2.3 & 0.164 & 0.164 & 92.6 \\
		&G& 1.43 & \bf 1.5 & 0.090 & 0.092 & 95.2 & \bf 0.7 & 0.088 & 0.092 & 95.2 & \bf -0.3 & 0.086 & 0.092 & 94.9 \\
		&P& 3.99 & \bf 1.9 & 0.828 & 0.822 & 94.6 & \bf 1.1 & 0.818 & 0.814 & 94.2 & \bf -0.3 & 0.801 & 0.799 & 93.5 \\[1ex]
		0.4 &C& 1.33 & \bf 2.8 & 0.216 & 0.212 & 93.8 & \bf 0.5 & 0.214 & 0.207 & 92.8 & \bf -2.5 & 0.212 & 0.214 & 92.4 \\
		&G& 1.67 & \bf 1.5 & 0.113 & 0.115 & 95.3 & \bf 0.5 & 0.111 & 0.114 & 95.0 & \bf -0.7 & 0.109 & 0.115 & 95.0 \\
		&P& 6.58 & \bf 1.6 & 1.313 & 1.308 & 94.8 & \bf 0.6 & 1.296 & 1.294 & 94.3 & \bf -1.1 & 1.267 & 1.27 & 93.6 \\[1ex]
		0.6 &C& 3.00 & \bf 0.0 & 0.373 & 0.379 & 94.4 & \bf -1.5 & 0.372 & 0.375 & 92.8 & \bf -3.7 & 0.368 & 0.401 & 93.2 \\
		&G& 2.50 & \bf 1.1 & 0.188 & 0.193 & 95.0 & \bf -0.01 & 0.186 & 0.192 & 94.5 & \bf -1.5 & 0.183 & 0.196 & 94.3 \\
		&P& 21.13 & \bf 0.2 & 3.898 & 3.916 & 94.4 & \bf -1.0 & 3.837 & 3.866 & 93.8 & \bf -3.1 & 3.745 & 3.79 & 92.5 \\[1ex]
		0.8 &C& 8.00 & \bf -2.9 & 0.845 & 1.001 & 95.2 & \bf -4.1 & 0.845 & 1.048 & 94.9 & \bf -5.8 & 0.835 & 1.175 & 95.2 \\
		&G& 5.00 & \bf -0.4 & 0.402 & 0.438 & 95.6 & \bf -1.7 & 0.399 & 0.443 & 95.0 & \bf -3.4 & 0.393 & 0.476 & 94.1 \\
		&P& 115 & \bf -4.3 & 19.454 & 20.246 & 92.2 & \bf -5.9 & 19.091 & 19.878 & 90.9 & \bf -8 & 18.728 & 19.531 & 89.0 \\
		\hline
		\hline
		$\tau$ & $\theta$& $\theta$ &  $\text{\bf PRB}_{\hat\theta^*}$ & $s_{\hat\theta^*}$ & $\overline{se}_{\hat\theta^*}$ & $\text{PC}_{\hat\theta^*}$ & $\text{\bf PRB}_{\tau}$ & $s_{\tau}$ & $\overline{se}_{\tau}$ & $\text{PC}_{\tau}$ & $\text{\bf PRB}_{\rho}$ & $s_{\rho}$ & $\overline{se}_{\rho}$ & $\text{PC}_{\rho}$ \\
		\hline
		0.1 &C& 0.22 & \bf 1.2 & 0.100 & 0.097 & 93.7 & \bf 3.6 & 0.119 & 0.118 & 94.3 & \bf 3.2 & 0.118 & 0.117 & 94.6 \\
		&G& 1.11 & \bf 0.2 & 0.054 & 0.058 & 94.6 & \bf 0.5 & 0.059 & 0.059 & 95.0 & \bf 0.5 & 0.058 & 0.058 & 95.4 \\
		&P& 1.56 & \bf 1.6 & 0.340 & 0.336 & 93.9 & \bf 1.6 & 0.344 & 0.337 & 94 & \bf 1.9 & 0.348 & 0.339 & 93.7 \\[1ex]
		0.2 &C& 0.50 & \bf 0.4 & 0.129 & 0.124 & 93.1 & \bf 1.8 & 0.148 & 0.147 & 94.4 & \bf 1.3 & 0.148 & 0.147 & 94.5 \\
		&G& 1.25 & \bf 0.2 & 0.070 & 0.073 & 94.9 & \bf 0.5 & 0.073 & 0.074 & 95.2 & \bf 0.3 & 0.074 & 0.073 & 94.9 \\
		&P& 2.48 & \bf 1.2 & 0.525 & 0.520 & 94.0 & \bf 1.9 & 0.552 & 0.542 & 93.4 & \bf 2.0 & 0.559 & 0.543 & 93.8 \\[1ex]
		0.3 &C& 0.86 & \bf -0.7 & 0.166 & 0.157 & 92.0 & \bf 0.8 & 0.185 & 0.184 & 94.5 & \bf 0.4 & 0.188 & 0.188 & 94.6 \\
		&G& 1.43 & \bf 0.2 & 0.088 & 0.091 & 94.6 & \bf 0.5 & 0.091 & 0.093 & 95.3 & \bf 0.5 & 0.094 & 0.094 & 94.7 \\
		&P& 3.99 & \bf 0.7 & 0.813 & 0.810 & 93.9 & \bf 1.9 & 0.876 & 0.860 & 94.3 & \bf 2.1 & 0.917 & 0.889 & 94.0 \\[1ex]
		0.4 &C& 1.33 & \bf -1.0 & 0.214 & 0.202 & 91.6 & \bf 1.0 & 0.240 & 0.234 & 94.1 & \bf 0.6 & 0.245 & 0.242 & 94.5 \\
		&G& 1.67 & \bf -0.04 & 0.110 & 0.114 & 95.0 & \bf 0.5 & 0.116 & 0.117 & 95.4 & \bf 0.5 & 0.117 & 0.116 & 95.1 \\
		&P& 6.58 & \bf 0.1 & 1.287 & 1.287 & 94.1 & \bf 2.0 & 1.465 & 1.434 & 94.3 & \bf 2.2 & 1.567 & 1.512 & 93.9 \\[1ex]
		0.6 &C& 3.00 & \bf -2.5 & 0.37 & 0.37 & 91.9 & \bf 0.7 & 0.418 & 0.409 & 93.9 & \bf 0.1 & 0.440 & 0.446 & 94.3 \\
		&G& 2.50 & \bf -0.7 & 0.185 & 0.194 & 94.4 & \bf 0.6 & 0.198 & 0.201 & 95.4 & \bf 0.2 & 0.206 & 0.202 & 94.7 \\
		&P& 21.13 & \bf -1.7 & 3.805 & 3.841 & 93.4 & \bf 2.6 & 4.842 & 4.723 & 94.6 & \bf 2.7 & 5.701 & 5.397 & 93.1 \\[1ex]
		0.8 &C& 8.00 & \bf -4.8 & 0.850 & 1.089 & 95.2 & \bf 0.8 & 0.975 & 0.943 & 94.0 & \bf -0.6 & 1.044 & 1.162 & 95.6 \\
		&G& 5.00 & \bf -2.6 & 0.398 & 0.461 & 94.8 & \bf 0.8 & 0.441 & 0.453 & 95.9 & \bf -0.4 & 0.466 & 0.445 & 92.7 \\
		&P& 115 & \bf -6.8 & 18.877 & 19.680 & 90.2 & \bf 3.1 & 28.605 & 27.771 & 94 & \bf 5.2 & 41.698 & 36.979 & 91.2 \\
		\hline
	\end{tabular}
	\renewcommand{\baselinestretch}{1}\small\normalsize
	\caption{\footnotesize Percentage relative bias (PRB), empirical standard deviation of the estimates ($s$), mean of the estimated standard errors ($\overline{se}$), and empirical percentage coverage (PC) of the approximate 95\% confidence interval for the dependence parameter. Estimates based on $5,000$ pseudo-random samples of size $n=200$.}
	\label{tb:bias_n200}
\end{table}

\begin{table}[htb]
	\renewcommand{\baselinestretch}{1.2}\small\normalsize
	\center
	\tiny
	\begin{tabular}{ccccccccccccccc}
		\hline
		$\tau$ & $\theta$& $\theta$ & $\text{\bf PRB}_{\hat\theta^c}$ & $s_{\hat\theta^c}$ & $\overline{se}_{\hat\theta^c}$ & $\text{PC}_{\hat\theta^c}$ & $\text{\bf PRB}_{\hat\theta^m}$ & $s_{\hat\theta^m}$ & $\overline{se}_{\hat\theta^m}$ & $\text{PC}_{\hat\theta^m}$ & $\text{\bf PRB}_{\hat\theta^M}$ & $s_{\hat\theta^M}$ & $\overline{se}_{\hat\theta^M}$ & $\text{PC}_{\hat\theta^M}$  \\
		\hline
		0.1 &C& 0.22 & \bf 6.2 & 0.072 & 0.071 & 94.6 & \bf 3.4 & 0.070 & 0.069 & 94.5 & \bf -0.07 & 0.068 & 0.068 & 94.2 \\
		&G& 1.11 & \bf 0.6 & 0.039 & 0.040 & 95.2 & \bf 0.3 & 0.038 & 0.040 & 95.3 & \bf -0.01 & 0.037 & 0.039 & 95.2 \\
		&P& 1.56 & \bf 1.2 & 0.233 & 0.237 & 95.3 & \bf 1.1 & 0.232 & 0.236 & 95.2 & \bf 0.8 & 0.23 & 0.234 & 95.1 \\[1ex]
		0.2 &C& 0.50 & \bf 4.2 & 0.092 & 0.090 & 94.9 & \bf 2.0 & 0.090 & 0.089 & 94.9 & \bf -0.9 & 0.089 & 0.089 & 94.5 \\
		&G& 1.25 & \bf 0.8 & 0.051 & 0.051 & 94.9 & \bf 0.4 & 0.050 & 0.051 & 95.1 & \bf -0.09 & 0.049 & 0.051 & 94.9 \\
		&P& 2.48 & \bf 1.3 & 0.362 & 0.368 & 95.0 & \bf 1.0 & 0.360 & 0.366 & 95.0 & \bf 0.5 & 0.356 & 0.363 & 94.8 \\[1ex]
		0.3 &C& 0.86 & \bf 2.5 & 0.116 & 0.116 & 94.7 & \bf 0.8 & 0.116 & 0.114 & 94.4 & \bf -1.5 & 0.115 & 0.116 & 93.8 \\
		&G& 1.43 & \bf 0.8 & 0.064 & 0.064 & 95.0 & \bf 0.4 & 0.063 & 0.064 & 94.8 & \bf -0.2 & 0.062 & 0.064 & 94.6 \\
		&P& 3.99 & \bf 1.3 & 0.565 & 0.574 & 95.1 & \bf 0.9 & 0.562 & 0.571 & 94.9 & \bf 0.1 & 0.556 & 0.566 & 94.8 \\[1ex]
		0.4 &C& 1.33 & \bf 1.7 & 0.149 & 0.149 & 94.3 & \bf 0.4 & 0.148 & 0.147 & 93.5 & \bf -1.5 & 0.147 & 0.150 & 93.6 \\
		&G& 1.67 & \bf 0.8 & 0.080 & 0.080 & 94.8 & \bf 0.3 & 0.079 & 0.080 & 94.8 & \bf -0.4 & 0.078 & 0.080 & 94.3 \\
		&P& 6.58 & \bf 1.1 & 0.903 & 0.915 & 95.2 & \bf 0.6 & 0.898 & 0.910 & 95.0 & \bf -0.3 & 0.887 & 0.901 & 94.7 \\[1ex]
		0.6 &C& 3.00 & \bf -0.1 & 0.264 & 0.263 & 94.3 & \bf -1.0 & 0.263 & 0.260 & 93.1 & \bf -2.3 & 0.262 & 0.274 & 93.6 \\
		&G& 2.50 & \bf 0.6 & 0.134 & 0.135 & 95.2 & \bf -0.01 & 0.133 & 0.134 & 95.0 & \bf -0.9 & 0.131 & 0.136 & 94.2 \\
		&P& 21.13 & \bf 0.4 & 2.722 & 2.744 & 94.7 & \bf -0.2 & 2.700 & 2.726 & 94.3 & \bf -1.4 & 2.663 & 2.695 & 93.8 \\[1ex]
		0.8 &C& 8.00 & \bf -1.8 & 0.607 & 0.663 & 94.7 & \bf -2.5 & 0.607 & 0.692 & 94.6 & \bf -3.5 & 0.603 & 0.772 & 94.9 \\
		&G& 5.00 & \bf -0.3 & 0.285 & 0.298 & 95.5 & \bf -1.0 & 0.283 & 0.301 & 95.0 & \bf -2.1 & 0.281 & 0.320 & 94.8 \\
		&P& 115 & \bf -2.2 & 13.974 & 14.173 & 93.5 & \bf -3.0 & 13.836 & 14.040 & 92.7 & \bf -4.3 & 13.651 & 13.865 & 91.4 \\
		\hline
		\hline
		$\tau$ & $\theta$& $\theta$ &  $\text{\bf PRB}_{\hat\theta^*}$ & $s_{\hat\theta^*}$ & $\overline{se}_{\hat\theta^*}$ & $\text{PC}_{\hat\theta^*}$ & $\text{\bf PRB}_{\tau}$ & $s_{\tau}$ & $\overline{se}_{\tau}$ & $\text{PC}_{\tau}$ & $\text{\bf PRB}_{\rho}$ & $s_{\rho}$ & $\overline{se}_{\rho}$ & $\text{PC}_{\rho}$ \\
		\hline
		0.1 &C& 0.22 & \bf 1.4 & 0.069 & 0.068 & 94.4 & \bf 2.7 & 0.082 & 0.083 & 95.6 & \bf 2.5 & 0.081 & 0.082 & 95.6 \\
		&G& 1.11 & \bf 0.1 & 0.038 & 0.040 & 95.3 & \bf 0.2 & 0.042 & 0.042 & 94.8 & \bf 0.2 & 0.041 & 0.041 & 94.9 \\
		&P& 1.56 & \bf 1.0 & 0.232 & 0.236 & 95.2 & \bf 0.8 & 0.232 & 0.235 & 95.3 & \bf 1.1 & 0.236 & 0.238 & 95.3 \\[1ex]
		0.2 &C& 0.50 & \bf 0.4 & 0.090 & 0.088 & 94.4 & \bf 1.4 & 0.102 & 0.103 & 95.5 & \bf 1.0 & 0.102 & 0.104 & 95.4 \\
		&G& 1.25 & \bf 0.1 & 0.050 & 0.051 & 95.0 & \bf 0.3 & 0.052 & 0.052 & 94.8 & \bf 0.1 & 0.052 & 0.052 & 94.9 \\
		&P& 2.48 & \bf 0.9 & 0.359 & 0.365 & 94.9 & \bf 1.0 & 0.377 & 0.381 & 94.9 & \bf 1.2 & 0.380 & 0.381 & 95.1 \\[1ex]
		0.3 &C& 0.86 & \bf -0.2 & 0.116 & 0.112 & 93.7 & \bf 0.6 & 0.127 & 0.130 & 95.5 & \bf 0.3 & 0.130 & 0.132 & 95.4 \\
		&G& 1.43 & \bf 0.1 & 0.063 & 0.064 & 94.8 & \bf 0.3 & 0.065 & 0.065 & 95.0 & \bf 0.3 & 0.067 & 0.067 & 94.2 \\
		&P& 3.99 & \bf 0.7 & 0.560 & 0.570 & 94.9 & \bf 1.1 & 0.597 & 0.601 & 95.1 & \bf 1.2 & 0.627 & 0.626 & 94.9 \\[1ex]
		0.4 &C& 1.33 & \bf -0.4 & 0.148 & 0.144 & 92.9 & \bf 0.7 & 0.165 & 0.165 & 95.0 & \bf 0.6 & 0.170 & 0.170 & 95.1 \\
		&G& 1.67 & \bf -0.03 & 0.079 & 0.080 & 94.7 & \bf 0.3 & 0.082 & 0.082 & 95.2 & \bf 0.3 & 0.083 & 0.082 & 94.6 \\
		&P& 6.58 & \bf 0.4 & 0.895 & 0.908 & 95.0 & \bf 1.0 & 1.004 & 1.004 & 94.6 & \bf 1.2 & 1.073 & 1.064 & 94.6 \\[1ex]
		0.6 &C& 3.00 & \bf -1.6 & 0.263 & 0.260 & 92.7 & \bf 0.3 & 0.288 & 0.288 & 95.1 & \bf 0.1 & 0.306 & 0.311 & 95.0 \\
		&G& 2.50 & \bf -0.4 & 0.132 & 0.135 & 94.6 & \bf 0.3 & 0.139 & 0.140 & 95.5 & \bf 0.1 & 0.144 & 0.144 & 95.3 \\
		&P& 21.13 & \bf -0.6 & 2.689 & 2.717 & 94.3 & \bf 1.5 & 3.319 & 3.303 & 94.6 & \bf 1.5 & 3.892 & 3.821 & 93.8 \\[1ex]
		0.8 &C& 8.00 & \bf -2.9 & 0.609 & 0.721 & 94.9 & \bf 0.5 & 0.680 & 0.664 & 93.9 & \bf -0.07 & 0.740 & 0.779 & 96.0 \\
		&G& 5.00 & \bf -1.5 & 0.282 & 0.311 & 95.1 & \bf 0.4 & 0.305 & 0.312 & 95.8 & \bf -0.3 & 0.329 & 0.322 & 92.2 \\
		&P& 115 & \bf -3.4 & 13.757 & 13.969 & 92 & \bf 1.5 & 19.564 & 19.373 & 94.4 & \bf 2.6 & 28.53 & 26.796 & 92.5 \\
		\hline
	\end{tabular}
	\renewcommand{\baselinestretch}{1}\small\normalsize
	\caption{\footnotesize Percentage relative bias (PRB), empirical standard deviation of the estimates ($s$), mean of the estimated standard errors ($\overline{se}$), and empirical percentage coverage (PC) of the approximate 95\% confidence interval for the dependence parameter. Estimates based on $5,000$ pseudo-random samples of size $n=400$.}
	\label{tb:bias_n400}
\end{table}


\newpage
\renewcommand{\baselinestretch}{1}\small\normalsize

\end{document}